\documentclass[acmsmall]{acmart}
\AtBeginDocument{%
  \providecommand\BibTeX{{%
    \normalfont B\kern-0.5em{\scshape i\kern-0.25em b}\kern-0.8em\TeX}}}
\setcopyright{acmcopyright}
\copyrightyear{2021}
\acmYear{2021}
\acmDOI{10.1145/3484197}
\acmJournal{TQC}
\acmVolume{2}
\acmNumber{3}
\acmArticle{12}
\acmMonth{9}

\usepackage{multirow}
\usepackage{amsmath}
\usepackage{pdfpages}
\usepackage{graphicx}
\usepackage{physics}
\usepackage{epstopdf}
\usepackage{mathtools}
\usepackage{amsthm}
\usepackage{float}
\usepackage{algorithmic}
\usepackage{cancel}
\usepackage[ruled,vlined]{algorithm2e}
\input{Qcircuit}
\DeclarePairedDelimiter{\ceil}{\lceil}{\rceil}

\newcommand{\HilD}{\mathcal{H}^D}
\newcommand{\Hild}{\mathcal{H}^d}

\newcommand{\Hildout}{\mathcal{H}^{d}_{out}}
\newcommand{\Hildperp}{\mathcal{H}^{d^{\perp}}}
\newcommand{\Hildperpo}{\mathcal{H}^{d^{\perp}}_{out}}

\newcommand{\E}{\mathcal{E}}

\newcommand{\T}{\mathcal{T}}
\newcommand{\Ti}{\mathcal{T}_{ideal}}
\newcommand{\A}{\mathcal{A}}
 
\newcommand{\kc}{\ket{\phi^c_i}}
\newcommand{\kr}{\ket{\phi^r_i}}
\newcommand{\krm}{\ket{\phi^r_i}^{\otimes M}}
\newcommand{\kp}{\ket{\phi^p_i}}
\newcommand{\kb}{\ket{\phi^b_i}}
\newcommand{\ke}{\ket{\phi^{\mathcal{E}}_i}}
\newcommand{\negl}{negl}
\newcommand{\nonnegl}{non\text{-}\negl}
\newcommand{\Sin}{S_{in}}
\newcommand{\Sout}{S_{out}}

\newcommand{\cpp}{hrv-id}
\newcommand{\qp}{lrv-id}
\newcommand{\cps}{hrv-id-swap}
\newcommand{\cpg}{hrv-id-gswap}

\newtheorem{theorem}{Theorem}

\newtheorem{definition}{Definition}

\begin{document}

\title{ Client-Server Identification Protocols with Quantum PUF}

\author{Mina Doosti}
\email{m.doosti@sms.ed.ac.uk}
\orcid{0000-0003-0920-335X}
\author{Niraj Kumar}
\email{nkumar@exseed.ed.ac.uk}
\author{Mahshid Delavar}
\email{mahshid.delavar@gmail.com}
\affiliation{%
  \institution{School of Informatics, University of Edinburgh}
    \streetaddress{10 Crichton St.}
    \city{Edinburgh}
    \country{United Kingdom}
}

\author{Elham Kashefi}
\email{ekashefi@exseed.ed.ac.uk}
\affiliation{%
  \institution{School of Informatics, University of Edinburgh}
    \streetaddress{10 Crichton St.}
    \city{Edinburgh}
    \country{United Kingdom;}
  \institution{CNRS, LIP6, Sorbonne Universit\'{e}, Paris}
  \streetaddress{4 place Jussieu}
  \city{Paris}
  \country{France}}

\begin{abstract}
Recently, major progress has been made towards the realisation of quantum internet to enable a broad range of classically intractable applications. These applications such as delegated quantum computation require running a secure identification protocol between a low-resource and a high-resource party to provide secure communication. In this work, we propose two identification protocols based on the emerging hardware secure solutions, the quantum Physical Unclonable Functions (qPUFs). The first protocol allows a low-resource party to prove its identity to a high resource party and in the second protocol, it is vice-versa. Unlike existing identification protocols based on Quantum Read-out PUFs which rely on the security against a specific family of attacks, our protocols provide provable exponential security against any Quantum Polynomial-Time adversary with resource-efficient parties. We provide a comprehensive comparison between the two proposed protocols in terms of resources such as quantum memory and computing ability required in both parties as well as the communication overhead between them.
\end{abstract}


\maketitle

\section{Introduction} 
The recent advances in developing the quantum internet have enabled a broad range of applications from simple secure communication all the way to delegated quantum computation, with no counterparts in classical networks  \cite{broadbent2016quantum,fitzsimons2017private,wehner2018quantum,QuantumPZoo, pirandola2019advances, diamanti2019demonstrating, kumar2019practically, unruh2013everlasting}.

For most of such applications, a key security feature is the ability of secure authentication which provides a central role in performing secure communications over untrusted channels \cite{alagic2017quantum, dulek2019secure, boneh2013quantum}. Amongst different types of required security features, including confidentiality and authentication of data, mutual entity authentication is a crucial, yet most neglected, aspect \cite{kang2018controlled}. Entity authentication also referred to as \emph{identification}, is a method to prove the identity of one party called \emph{prover} to another party called \emph{verifier}. 
The focus of this work is to propose resource-efficient solutions for the purpose of mutual entity authentication between two parties in a quantum network by exploring the advantages of quantum communication. We consider both complementary scenarios where either the trusted verifier or a potentially malicious prover has limited resources in the identification protocol. 
To better motivate the two scenarios, consider the quantum cloud service platforms that are commercially available today \cite{arute2019quantum, cross2018ibm, computing2019pyquil, bergholm2018pennylane, blinov2021comparison}. In the first setting, a client with a low quantum resource (such as the one defined in \cite{broadbent2009universal}) wishes to identify a high-resource quantum centre that they perhaps have had a previous contract with, before proceeding to access their platform and load its sensitive data. In the complimentary setting, the quantum cloud provider wishes to verify the identity of its customer possessing low quantum resources before providing them with access. This asymmetry between the verifier and the prover calls for `party resource-specific' identification protocols which exploit this asymmetry to enhance the efficiency.

Among the recent works, Physical Unclonable Functions (PUFs) have emerged as cost-efficient, low-resource hardware tokens to achieve entity authentication \cite{delvaux2017security,herder2014physical, vskoric2012quantum, nikolopoulos2017continuous}. 
A PUF device utilises the random physical disorders that occur during the manufacturing process to provide security features. This randomness provides the desired high min-entropy and unpredictability features, and hence the PUF does not rely on extra cryptographic properties in the device \cite{herder2014physical,armknecht2016towards}. Assessing information from a PUF involves querying the device with a `challenge' (for example an electrical signal or an optical pulse) and obtaining a recognizable `response'. This response should be robust for a particular PUF device but highly variable for different PUFs in a way that for an adversary, each device seems to output a completely random response.  An example of a PUF is an optical glass slab with an inhomogeneous refractive index such that shining a laser-pulse with a fixed frequency and angle of incidence, results in the output pulse with fixed (or very less divergent) frequency. However, another glass slab with a slight difference in the distribution of index of refraction results in the output pulse with different characteristics for the same incident light \cite{pappu2002physical}. This uniqueness in the challenge-response pair for a particular PUF is the core feature in realising entity authentication and other cryptographic functionalities. 
Other hardware realisations of PUF include SRAM PUF, Ring Oscillator PUF and Arbiter PUF \cite{guajardo2007fpga,gassend2002silicon,suh2007physical}. However, recent cryptanalysis has shown that conventional PUF hardware devices do not provide rigorous security guarantees as anticipated and the unpredictability feature is compromised by modelling attacks \cite{ruhrmair2010modeling,ganji2016strong,ruhrmair2010modeling,khalafalla2019pufs}.

Some of these security issues are overcome with the recently proposed PUFs that utilise the properties of quantum mechanics \cite{arapinis2021quantum,vskoric2010quantum,goorden2014quantum,nikolopoulos2017continuous, gianfelici2020theoretical}.  Referred to as quantum PUF, or qPUF, these are completely positive trace-preserving operations that are accessed via sets of unique challenge-response pairs which are quantum states. Implementations of these devices include optical qPUF \cite{nikolopoulos2017continuous,goorden2014quantum}. One major advantage of qPUFs compared to previous PUF proposals is that apart from the high-min entropy of the qPUF device, the challenges and responses also exhibit high-min entropy due to the unclonability property in quantum mechanics \cite{wootters1982single}. This extra feature is non-existent in previous PUFs since the challenges and responses being classical states, can be perfectly cloned.  Hence it serves as a great motivation to study qPUF resource and security performance in achieving various cryptographic functionalities. Our current work provides two proposals for achieving identification using qPUFs.  \par
 
Intending to perform low-cost secure identification of the prover by the verifier using qPUF, we give a categorisation of the resources into three major segments. First is the `memory resource' which quantifies the type and amount of resources that a party possesses. It can either be a classical memory that we label as low cost or a quantum memory which is high cost since such a memory tends to be highly fragile and dissipative to the environment \cite{lvovsky2009optical}. Second is the `computing ability' resource which indicates the kind of operations a given party has the ability to perform. We denote a party with high computing ability as the one that can perform any bounded quantum polynomial quantum circuit operations \cite{watrous2003complexity}, and a low ability party as the one who is restricted to generation and measurement of quantum states in certain basis. And the third resource is the type and number of `communication rounds' required between the parties to establish identification. Often it is not possible to devise an identification scheme which minimises all the three types of resources for both the involved parties without compromising the underlying security. Hence, in this work, we propose two qPUF based identification schemes which achieve similar security guarantees but are vastly different in terms of the resource requirement for the involved parties. This allows the flexibility to deploy either of these schemes depending on individual constraints. \par

Our first proposal is a secure qPUF-based device identification protocol which requires the prover to only have access to the valid qPUF device without the requirement for any quantum memory or quantum computational resource, while the verifier is required to possess a local quantum database and the ability to perform quantum operations. This covers the scenario presented before where a quantum cloud provider wants to identify its customer. This type of qPUF-based identification protocols has been previously studied with different qPUF formalism \cite{vskoric2012quantum,nikolopoulos2017continuous}. In our work, we follow the formal definitions of a qPUF as proposed in \cite{arapinis2021quantum} which assumes that a qPUF is modelled by an unknown unitary operation of exponential size i.e. none of the involved parties, with polynomial resources, have a complete description of the device. This property of qPUF necessitates the use of a quantum distinguishing test in the protocol since the resulting response stats of the qPUF device are unknown states \cite{montanaro2013survey,buhrman2001quantum, chabaud2018optimal}.  This is in contrast with the previous quantum identification proposals, where some knowledge of the quantum operation was implicitly assumed to be known the parties, thus not necessitating the use of quantum distinguishability tests. However, this extra information allows proving the security against an only specific type of adversarial attacks. Our work generalises to provide exponentially high security against any quantum polynomial-time (QPT) adversary.

Our second proposal is a qPUF based protocol where the prover has a high computational resource, while, the verifier runs a purely classical algorithm, hence does not require to perform quantum operations. The verifier is however required to possess a local quantum database. This protocol can enable an almost classical client, to identify a quantum server in a quantum network. This protocol has a major advantage compared to the previous protocol that requires only one-way quantum communication. Construction of this protocol has taken inspiration from the ideas of blind quantum computing \cite{broadbent2009universal} to introduce the idea of randomly placing trap quantum states in-between the valid states. This, coupled with the unknown property of qPUF device provides provable security against any QPT adversary. \\

\noindent\textbf{Related Works:}
The idea of taking advantage of quantum communication between the verifier and the prover in PUF-based identification protocols was first introduced by Skoric in \cite{vskoric2010quantum}.  He defined the concept of \emph{quantum read-out of PUF (QR-PUF)} and designed an identification protocol based on it. The security of this protocol has been proved against special kinds of attacks including intercept-resend \cite{vskoric2010quantum,vskoric2012quantum}, Challenge Estimation \cite{vskoric2016security} and Quantum Cloning \cite{yao2016quantum} attacks. The practical realization of this protocol was shown by Goorden et al. \cite{goorden2014quantum}. 
In another work, Nikolopoulos and Diamanti introduced a different setup for QR-PUF-based identification protocol in which classical data is encoded to the continuous quadrature components of the quantized electromagnetic field of the probe \cite{nikolopoulos2017continuous}. 
The security of this scheme has also been proved in \cite{nikolopoulos2018continuous,fladung2019intercept} against a bounded adversary who can only prepare and measure the quantum states.
The common feature of the mentioned protocols \cite{vskoric2010quantum,nikolopoulos2017continuous} is full or partial knowledge of the verifier from the unitary modelling the QR-PUF. Recently, Arapinis et al. \cite{arapinis2021quantum} have introduced a novel notion of PUF, called qPUF. According to their definition, unlike the QR-PUFs and the same as classical PUFs, no one even the manufacturer and the verifier has no knowledge about the unitary of qPUF. This requirement leads to provable security of qPUFs against forgery attacks. Due to the considerable security features of qPUFs, we propose our identification protocols based on this kind of PUFs.
The main advantage of our proposals over the previous ones is their provable security against the most general form of attacks considering a QPT adversary. 
The other related works in the context of quantum related PUFs are \cite{gianfelici2020theoretical} and \cite{young2019quantum} where the former presents a theoretical framework for QR-PUF and the later is a different type of PUF based on quantum mechanics laws.

\section{Preliminaries} \label{sec:prelims}

This section presents the different ingredients required to construct a secure qPUF-based authentication scheme.

\subsection{Quantum Physical Unclonable Functions}

A quantum PUF, or qPUF, is a secure hardware cryptographic device which utilises the property of quantum mechanics \cite{arapinis2021quantum}. Similar to a classical PUF \cite{armknecht2016towards}, a qPUF is assessed via challenge and response pairs (CRP). However, in contrast to a classical PUF where the CRPs are classical states, the qPUF CRPs are quantum states. 

A qPUF manufacturing process involves a quantum generation algorithm, `QGen', which takes as an input a security parameter $\lambda$ and generates a PUF with a unique identifier \textbf{id},

\begin{equation}
    \text{qPUF}_{\textbf{id}} \leftarrow \text{qGen}(\lambda)
\end{equation}

Next we define the mapping provided by $\text{qPUF}_{\textbf{id}}$ which takes any input quantum state $\rho_{in} \in \mathcal{H}^{d_{in}}$ to the output state $\rho_{out} \in \mathcal{H}^{d_{out}}$. Here $\mathcal{H}^{d_{in}}$ and $\mathcal{H}^{d_{out}}$ are the input and output Hilbert spaces respectively corresponding to  the mapping that $\text{qPUF}_{\textbf{id}}$ provides as illustrated in Fig.~\ref{fig:puf}. This process is captured by the `qEval' algorithm which takes as an input a unique $\text{qPUF}_{\textbf{id}}$ device and the state $\rho_{in}$ and produces the state $\rho_{out}$,

\begin{equation}
    \rho_{out} \leftarrow \text{qEval}(\text{qPUF}_{\textbf{id}}, \rho_{in})
\end{equation}

A qPUF is labelled secure if it satisfies a few necessary requirements. The first property, \textbf{robustness}, ensures that if the qPUF is queried separately with two input quantum states $\rho_{in}$ and $\sigma_{in}$ that are $\delta_r$-indistinguishable to each other, then the output quantum states $\rho_{out}$ and $\sigma_{out}$ must also be $\delta_r$-indistinguishable. More formally we have:\\

\noindent\textbf{$\delta_r$-Robustness} On any two input states $\rho_{in}$ and $\sigma_{in}$ that are $\delta_r$-indistinguishable, the corresponding output quantum states $\rho_{out}$ and $\sigma_{out}$ are also $\delta_r$-indistinguishable with overwhelming probability,
\begin{equation}
 \mathrm{Pr}[\delta_r\le F(\rho_{out}, \sigma_{out})\le 1] \geqslant 1 - \negl(\lambda).
\end{equation}


where $\negl(\lambda)$ is a negligible quantity dependent on the desired security parameter, and the probability is taken over the choice of all the states. We say two quantum states $\rho$ and $\sigma$ are $\delta$-indistinguishable if $\delta \leqslant F(\rho, \sigma) \leqslant 1$, where $F(\rho, \sigma) = \text{Tr}\sqrt{\sqrt{\rho}\sigma\sqrt{\rho}}$ is the fidelity distance measure between the quantum states. Alternatively, other distance measures such as trace norm, euclidean norm (any shatten-p norm) can also be used to define security requirements for qPUF. 

The second property, \textbf{collision resistance}, ensures that if the same qPUF is queried separately with two input quantum states $\rho_{in}$ and $\sigma_{in}$ that are $\delta_c$-distinguishable, then the output states $\rho_{out}$ and $\sigma_{out}$ must also be $\delta_c$-distinguishable with an overwhelmingly high probability, more preceisely,\\

\noindent\textbf{$\delta_c$-Collision-Resistance (Strong)} For any qPUF on any two input states $\rho_{in}$ and $\sigma_{in}$ that are $\delta_c$-distinguishable, the corresponding output states $\rho_{out}$ and $\sigma_{out}$ are also $\delta_c$-distinguishable with overwhelming probability,
\begin{equation}
\mathrm{Pr}[0\le F(\rho_{out}, \sigma_{out})\le 1-\delta_c] \geqslant 1 - \negl(\lambda).
\end{equation}


We say two quantum states $\rho$ and $\sigma$ are $\delta$-distinguishable if $0 \leqslant F(\rho, \sigma) \leqslant  1 - \delta$. The parameters $\delta_r$ and $\delta_c$ are determined by the security parameter $\lambda$. The properties defined above are crucial for the correctness of secure systems composed of qPUFs. Also for qPUFs, the condition $\delta_c \leqslant 1 - \delta_r$ must be satisfied to characterise a desired qPUF.

All the above properties can be satisfied by a unitary map i.e. if $\text{qPUF}_{\textbf{id}}^{\dagger}\text{qPUF}_{\textbf{id}} = \mathbf{I}$, where $\mathbf{I}$ is an identity matrix. As a consequence, here we consider the qPUF construction to be a unitary matrix $U \in \mathbb{C}^{D \times D}$, where $D = d_{in} = d_{out}$. \footnote{Other CPTP maps that attach an ancilla such that $d_{out} > d_{in}$ also satisfy all the properties. We do not consider such maps for the construction of PUFs. This could however be an interesting line of extension of PUFs.}

\begin{figure}[ht!]
\includegraphics[scale=0.45]{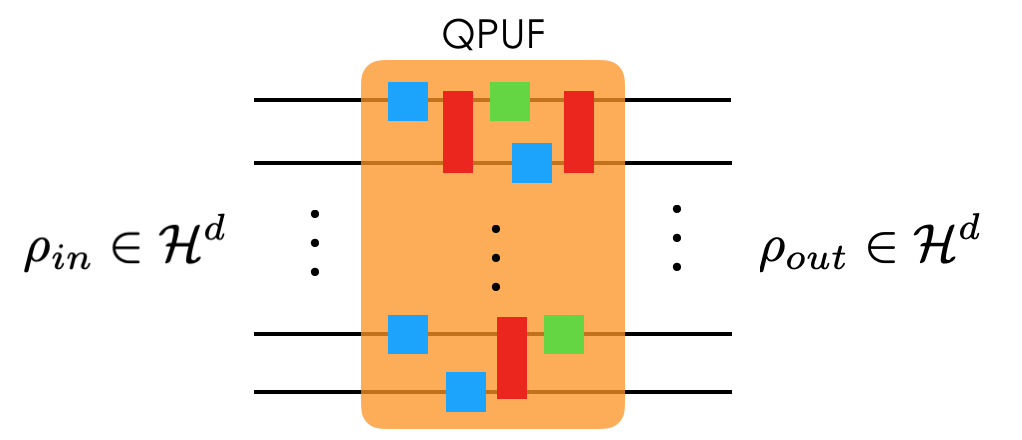}
\centering
\caption{Illustration of qPUF as a unitary operation with input and output quantum states in $\mathcal{H}^d$. The blue and green boxes are single-qubit gates, while red boxes are two-qubit gates. These are the building blocks for the qPUF construction.}
\label{fig:puf}
\end{figure}

A crucial security feature of the qPUF device is the unforgeability property. It states that estimating the response of the device with high enough fidelity when a challenge is picked uniformly at random from the Haar measure states is exponentially unlikely without possessing the device. Formally this means that for a challenge state $\rho_{in} \in \mathcal{H}^{D}$,
\begin{equation}
    \text{Pr}[F(\sigma, \rho_{out}) \geqslant 1 - \mu| \rho_{in} \in \mathcal{H}^{D}] \leqslant \negl(\log D)
\end{equation}
where $\sigma$ is the optimal response generated to a given challenge $\rho_{in}$, $\rho_{out}$ is the response generate by qPUF device on the given challenge and $\mu = \negl(\log D)$. We note that the security parameter of qPUFs is related to their dimension. For instance, for an n-qubit qPUF, with a unitary of size $D=2^n \times 2^n$, the security parameter is usually selected to be $n$.
%
\subsection{Quantum Adversarial Model and Security Definitions}
Strong notions of the security of quantum cryptographic proposals require cryptanalysis against adversaries which also possess quantum capabilities of varying degree \cite{boneh2011random, mosca2018cybersecurity, song2014note}.
The strongest such notion is achieved by assuming no restrictions on the adversary's computational power and resources. This security model, also known as security against \emph{unbounded adversary}, is usually too strong to be achieved by most cryptographic primitives such as qPUFs.
It has been shown in~\cite{arapinis2021quantum}, that unitary qPUFs cannot remain secure against an unbounded adversary. Thus the standard security model that we also use in this paper is the notion of security against efficient quantum adversaries or in other words quantum polynomial time (QPT) adversaries. We define such an adversary attack in the context of qPUFs. A QPT adversary with query access to the qPUF is defined as an adversary that can query the qPUF oracle with polynomially many (in the security parameter) challenges and has polynomial sized quantum register to store the quantum CRPs. The QPT adversary is also allowed to run any efficient quantum algorithm in the class BQP. This is the quantum analogue of \textit{Chosen Message Attack (CMA)} model in the classical cryptography, where the adversary is allowed for querying the primitive with messages of his choice in an adaptive way. The security of most qPUF-based cryptographic protocols relies on the unforgeability property of qPUF which is described previously.

Here we follow the same definitions of \emph{existential} and \emph{selective} unforgeability defined in~\cite{arapinis2021quantum} and restate them as follows:

\begin{enumerate}
    \item \emph{Existential unforgeability}: A qPUF satisfies \emph{existential unforgeability} if having access to a register $S$ containing a polynomial number pairs of challenges ,selected by the adversary, and their respective responses from qPUF, the probability that any QPT adversary $\A_{\text{QPT}}$ chooses a quantum challenge $\rho_{in}$ which is $\mu$-distinguishable from all challenges $S$, and successfully generates a response $\sigma$  which is $\epsilon$-indistinguishable from the valid qPUF's response $\rho_{out}$, is bounded by a negligible function of the security parameter. In other words, no QPT adversary can generate even a single \emph{valid} new quantum challenge-response pair with non-negligible probability,
     \begin{equation}
     \begin{split}
         \text{Pr}\big[(\rho_{in}, \sigma) \wedge F(\sigma, \rho_{out}) \geqslant 1 - \epsilon\big| F(\rho_{in},\rho) \leqslant 1 - \mu \hspace{1mm}, \forall \rho \in S_{in}\big] \leqslant \negl(\lambda)
         \end{split}
     \end{equation} 
    where $S_{in}$ is the set of all challenges in the $S$ register. 

    \item \emph{Selective unforgeability}:  A qPUF satisfies \emph{selective unforgeability} if having access to a register $S$ containing a polynomial number pairs of challenges ,selected by the adversary, and their respective responses from qPUF, the average probability that any QPT adversary $\A_{\text{QPT}}$ receives a quantum challenge chosen uniformly at random $\rho_{in}$, 
    and successfully generates a response $\sigma$ which is $\epsilon$-indistinguishable from the valid qPUF's response $\rho_{out}$, is bounded by a negligible function of the security parameter. In other words, no QPT adversary can generate \emph{valid} quantum responses for randomly selected challenges, on average with non-negligible probability,
     \begin{equation}
     \begin{split}
         \underset{\rho_{in} \in \HilD}{\text{Pr}}\big[F(\sigma, \rho_{out}) \geqslant 1 - \epsilon \big] \leqslant \negl(\lambda)
         \end{split}
     \end{equation} 
      where $\HilD$ is the Hilbert space from which the challenges are being picked uniformly according to the Haar measure.
\end{enumerate}
Note that in both the attack models, we allow for the possibility of adaptive kinds of attacks from the adversary \cite{armknecht2016towards}.
The results in~\cite{arapinis2021quantum} shows that a unitary qPUF cannot satisfy existential unforgeability against QPT adversaries. This is due to the existence of a quantum emulation based algorithm which states that picking a new challenge $\rho_{in}$ in the subspace spanned by the challenges in $S$ register such that $\rho_{in}$ is $\mu$-distinguishable from all the challenges in $S$, it is efficiently possible to output a response state $\sigma$ such that $F(\sigma, \rho_{out}) \approx 1$.   qPUFs however do satisfy 
selective unforgeability \cite{arapinis2021quantum}. Their result states that the success probability of any QPT adversary to output the response of a Haar random challenge state $\rho \in \HilD$ with non-negligible fidelity is bounded by:

\[\text{Pr}_{success} = \underset{\mathcal{A}_{QPT}}{\max}\hspace{1mm} \underset{\rho\in\HilD}{\text{Pr}}[F(\sigma, U\rho U^{\dagger}) \geqslant \delta] \leqslant \frac{d+1}{D}\]

where $\Sin$ is the set of challenges in the $S$ register and $d$ is the dimension of the challenge subspace known to the $\mathcal{A}_{QPT}$ via the $\Sin$ register. $D$ is the size of the qPUF unitary and $\delta$ is a negligible function in poly$(\log(D))$.
In our work, we assume the qPUF is an unknown unitary transformation. This assumption allows us to use the qPUF as a selectively unforgeable device according to the above definition. We restate the proof of qPUF unforgeability in the Appendix~\ref{sec:sel-unf}. Moreover, another quantum toolkit that we use for our protocol is Equality testing of the quantum state, which is to test whether two \emph{unknown} quantum states are the same. This is a well-studied topic and we describe the optimal quantum protocols for Equality testing in the Appendix~\ref{sec:itest}.
%
\subsection{General Description of device-based identification protocol}\label{sec:protocol}

An identification protocol, also called a device-authentication protocol, is run between a verifier and a prover. A verifier's task is to check the identity of the prover by identifying whether the prover is the correct owner of a valid device. Our setting assumes that the verifier and the prover having a valid device behave honestly. The security is provided against an adversary who has had limited access to the valid device in the past and currently does not possess the valid device. Based on the limited knowledge that the adversary has, their objective is to successfully impersonate themselves as the valid owner of the device. Prior to providing the details of the construction of device identification protocols using qPUF, we describe a common structure in these protocols. Any such protocol consists of three sequential phases: \emph{setup phase} (or enrollment phase), \emph{identification phase} and \emph{verification phase} \cite{nikolopoulos2017continuous,vskoric2010quantum,pappu2002physical}.

\begin{enumerate}
    \item \emph{Setup phase}: A setup phase is the beginning phase of the protocol. Here the verifier has the valid device (In this case a PUF/qPUF) and locally prepares a database consisting of multiple challenge and response pairs of this device. The challenges and responses, namely Challenge-Response pairs (CRPs) are stored in the verifier's local database. For protocols we define over the next sections, we assume that the verifier's quantum capabilities are restricted to quantum polynomial time. Hence the size of verifier's database can only be polynomial while the device itself is of exponential size. Once the local database is generated, the device is physically transferred to the prover over a public channel.
    
    \item \emph{Identification phase}: The setup phase is followed by the identification phase where the verifier sends one or multiple challenges, usually chosen at random, to the prover from the CRP database. The challenge(s) is sent over a public (quantum) channel to the prover.

    The prover who has the valid device obtains the responses of the received challenges by interacting them with the device and produces the response. Then the prover sends either the response directly, or sends some classical or quantum information related to the response to the verifier. We note that qPUF-based identification protocols would mostly differ in this phase by varying the number of challenges sent to the prover and the type of information received by the verifier. 
    
    \item \textit{Verification phase}: In the verification phase, the verifier runs a quantum or classical verification algorithm on the information received from the prover. We denote that the verifier correctly identifies the prover if the verification algorithm outputs 1. Otherwise, it aborts. 
\end{enumerate}

The \textbf{Correctness} or \textbf{Completeness} of an identification protocol is defined as the success probability of an honest prover over $n$ rounds of identification, in the absence of any adversary or noise, should be 1. The \textbf{Soundness} of an identification protocol insures that the success probability of any adversary (depending on the adversarial model) in passing the verification phase over the $n$ rounds of identification, should be negligible in the security parameter. 

\section{qPUF identification protocol with high-resource verifier}
\label{sec:cp}
An identification protocol is run between a verifier and a prover where the verifier is tasked with correctly identifying the prover who owns the device. Our setting assumes that the verifier and device owner behave honestly. The security is provided against an adversary who has limited access to the device only in the pre-protocol phase and her objective is to be identified as the valid device owner. We propose the construction of two identification protocols using qPUFs which provide exponential security against any QPT adversary. qPUF is described by an unknown unitary transformation in $\mathcal{H}^{D}$ whose construction is defined in Supplementary material.

The first qPUF-based device identification protocol we propose is the quantum analogue of the standard PUF-based identification scheme between the verifier (Alice) and the prover (Bob)~\cite{ruhrmair2014pufs,delavar2017puf} as shown in Figure~\ref{fig:qid-p1}. Prior to detailing the protocol, we list its salient features,
\begin{itemize}
    \item The prover is not required to have quantum memory as well as computing ability resource\footnote{Here we note that the prover applies the qPUF transformation on the challenge states by interacting them with the device. Nevertheless, we do not consider this as the computing ability of the prover and by no computing ability we refer to the fact that the prover does not need to run any extra quantum computations.}, whereas the verifier is required to have high quantum memory and high computing ability resource (restricted to QPT memory and computation).
    \item The protocol requires a 2-way quantum communication link between the prover and verifier.
    \item The protocol has a quantum verification phase i.e. the prover sends information in quantum states to the verifier who then performs a verification test to certify if the device is valid.
    \item The protocol provides perfect completeness and an exponentially-high security guarantee against any adversary with QPT resources. 
\end{itemize}

\begin{figure}[t]
\includegraphics[scale=0.55]{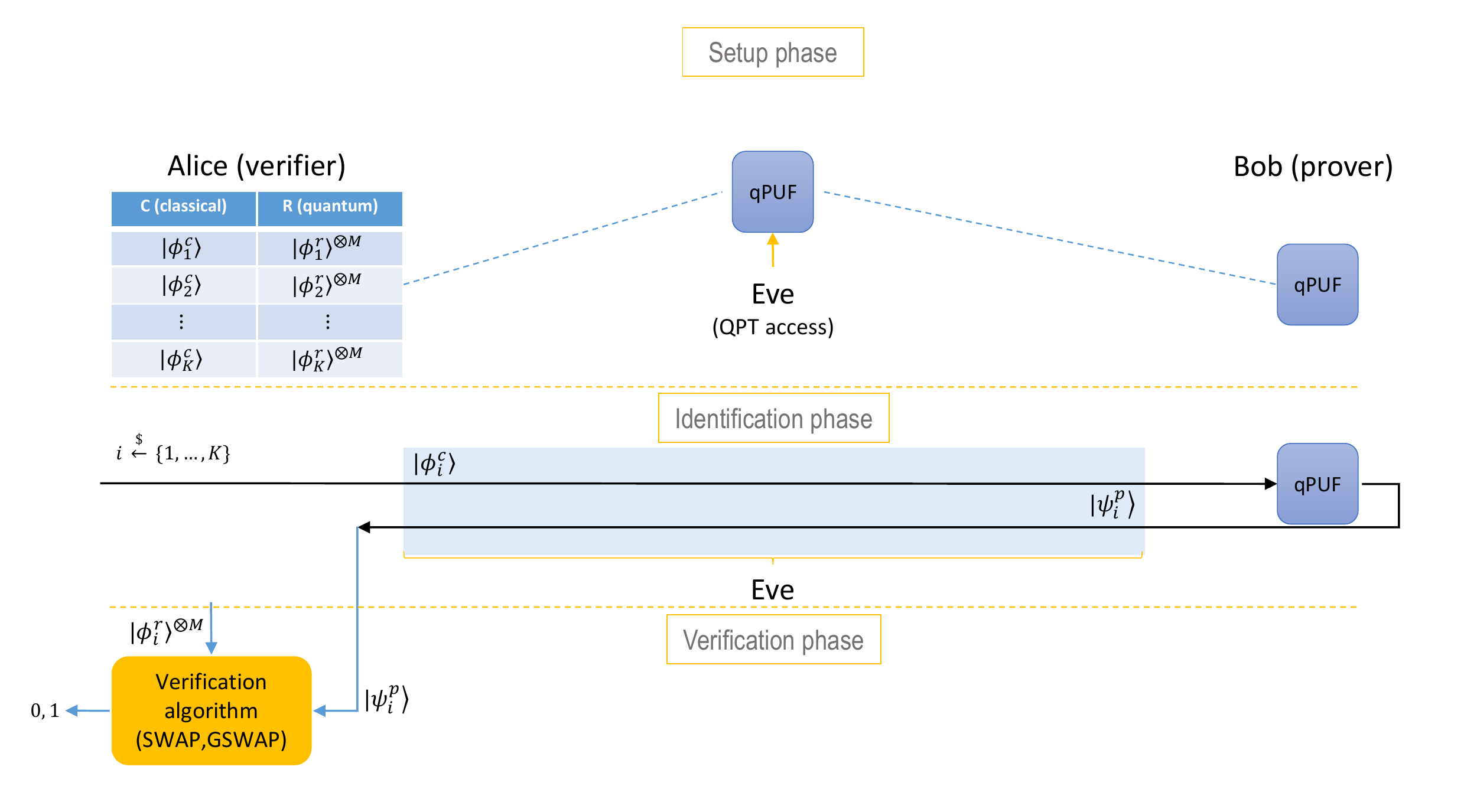}
\centering
\caption{qPUF-based identification protocol with high-resource verification between Alice(verifier) and Bob(prover) (\cpp). The protocol is divided into three sequential phases, \emph{setup phase}, \emph{identification phase}, and \emph{verification phase}. The protocol is analysed in presence of a QPT adversary Eve which can gain information about the device during the \emph{setup phase} and \emph{identification phase}. In the last phase, Alice runs a quantum verification algorithm and outputs a classical bit `1' if Bob's device is correctly identified. Otherwise, she outputs `0'. }
\label{fig:qid-p1}
\end{figure}

\subsection{Protocol description}
This protocol, referred as \cpp, is run between the Alice, the verifier, and Bob, the prover and it is divided into three sequential phases,
\begin{enumerate}
    \item \emph{Setup phase}:
            \begin{enumerate}
                \item Alice has the qPUF device. 
                \item She randomly picks $K \in \mathcal{O}(\text{poly} \log D)$ classical strings $\phi_i \in \{0,1\}^{\log D}$.
                \item Alice selects and applies a Haar-random state generator operation denoted by the channel $\E$ to locally create the corresponding quantum states in $\HilD$: $\phi_i \overset{\E}{\rightarrow} \kc,\hspace{2mm} \forall i \in [K]$.
                \item She queries the qPUF individually with each challenge $ \kc$ a total of $M$ number of times to obtain $M$ copies of the response state $\kr$ and stores them in her local database $S \equiv \{\kc, \krm\}_{i=1}^{K}$. 
                \item Alice publicly transfers the qPUF to Bob.
            \end{enumerate}
            
            To be able to investigate the security in a strong and general setting, we do not assume the qPUF's transition of being secure, in the sense that any QPT adversary Eve is allowed to query the qPUF during transition an $\mathcal{O}(\text{poly} \log D)$ number of times and thus build its local database. Due to the conditions on the selective unforgeability of the qPUF (Appendix~\ref{sec:sel-unf}), it is important that Alice picks her challenges $\kc \in S$ at random from a distribution over the Hilbert space $\HilD$. This, in turn, implies that the encoding unitary operation $\E$ is a haar random unitary \cite{arapinis2021quantum}. We note that an alternate efficient simulation of $\E$ was proposed by \cite{alagic2020efficient}.
    \item \emph{Identification phase}:
            \begin{enumerate}
                \item Alice uniformly selects a challenge labelled ($i \xleftarrow{\$} [K]$), and sends the state $\kc$ over a public quantum channel to Bob.
                \item Bob generates the output $\kp$ by querying the challenge received from Alice to the qPUF device.
                \item The output state $\kp$ is sent to Alice over a public quantum channel.
                \item This procedure is repeated with the same or different states a total of $R \leq K$ times. 
            \end{enumerate}

    \item \emph{Verification phase}:
            \begin{enumerate}
                \item Alice runs a quantum equality test algorithm on the received response from Bob and the $M$ copies of the correct response that she has in the database. This algorithm is run for all the $R$ CRP pairs.
                \item She outputs `1' implying successful identification if the test algorithm returns `1' on all CRPs. Otherwise, she outputs `0'. 
            \end{enumerate}
    Sections~\ref{sec:swapver} and \ref{sec:gswapver} describe the quantum verification algorithm run by Alice.          
\end{enumerate}

For this protocol, we define the security in terms of completeness and soundness properties. Completeness of {\cpp} protocol is the probability that Alice outputs `1' in the verification phase in absence of an adversary Eve. This implies that the verification algorithm must output `1' for all the $R$ rounds of the protocol with a probability that differs negligibly in the security parameter from 1,
\begin{equation}
\small
     \text{Pr}[\text{Ver accept}_{\text{H}}] =  \text{Pr}\big[\prod_{i=1}^{R}\texttt{(qVer}(\kp, \kr) = 1)\big] = 1 - \negl(\lambda)
     \label{eq:complete1}
\end{equation}
where the subscript H denotes the honest device holder.
%
Soundness of the protocol is defined as the probability that a QPT Eve passes the verification test of Alice.  We say the {\cpp} is sound (or secure) if this probability is negligible in the security parameter:
\begin{equation}
\small
    \text{Pr}[\text{Ver accept}_{\text{Eve}}] =  \text{Pr}\big[\prod_{i=1}^{R}\texttt{(qVer}(\rho_i, \kr) = 1)\big] =  \negl(\lambda)
    \label{eq:gensound1}
\end{equation}
where $\rho_{i}$ is the state sent by Alice in the $i$-th round.

Since our protocol is based on qPUF as defined in \cite{arapinis2021quantum}, Alice has no knowledge about the unitary of qPUF except the database $S$ she can obtain by querying. Consequently, her responses in $S$ are unknown quantum states. This calls for quantum equality test based verification algorithms to enable her to validate the received states. We investigate the optimal one-sided error test, the SWAP test \cite{buhrman2001quantum} and its generalised version (GSWAP)\cite{chabaud2018optimal} as two well-studied and practical quantum equality tests. These tests are described in Supplementary Material. 

\subsection{Verification with SWAP test} \label{sec:swapver}

The first proposal for Alice's \texttt{qVer} algorithm is the SWAP test and the identification protocol using this test is called {\cps}.
Its single run inputs one copy of each received state and Alice's response state and produces a binary outcome to probabilistically determine the equality between two states. A single run, however, does not provide a low enough test error rate. To obtain an exponentially low rate, the test is repeated $M$ number of times for the same challenge state where $M$ is proportional to the inverse-log of the desired error probability. The error can be further lowered by choosing $N \leq K$ distinct challenge states such that the test is run for  $R = N\times M$ number of times and the prover is successfully identified only if he passes all the runs. In the next two theorems, we show that SWAP based test algorithm provides us with the desired completeness and soundness properties required in the protocol.

\begin{theorem}[\textbf{SWAP Completeness}]\label{th:swap-comp} In absence of Eve, the probability that Bob's response state generated from the valid qPUF $\kp = qPUF(\kc)$ passes all the $R$ {\normalfont SWAP} test runs is,
{\normalfont 
\begin{equation}
\small
    \text{Pr}[\text{Ver accept}_{\text{H}}] =  \text{Pr}\big[\prod_{i=1}^{R}\text{(SWAP}(\kp, \kr) = 1)\big] = 1
\end{equation}}
\end{theorem}

\begin{proof}
When Alice receives Bob's response $\kp$ which is generated from the valid qPUF device for all the $i \in [R]$ copies of the challenge state, then $\kp = \kr$. This implies that $F(\kp,\kr) = 1$ for all $i \in [R]$. From Eq~\ref{eq:swapaccept}, we see that,
\begin{equation}
\text{Pr}\big[(\text{SWAP}(\kp, \kr) = 1] = \frac{1}{2} + \frac{1}{2}F^2(\kp,\kr) = 1, \hspace{2mm} \forall i \in [R]  
\end{equation}
Since in the honest setting, the states received from Bob over $R$ rounds are all valid qPUF pure states which are unentangled to each other, hence the SWAP tests for all the $R$ rounds are independent tests. This implies that,
\begin{equation}
        \text{Pr}[\text{Ver accept}_{\text{H}}] =  \text{Pr}\big[\prod_{i=1}^{R}\text{(SWAP}(\kp, \kr) = 1)\big] = \prod_{i=1}^{R}\text{Pr}\big[\text{SWAP}(\kp, \kr) = 1\big] = 1
\end{equation}
This completes the proof.
\end{proof}

To characterise the soundness, we bound Eve's success probability in passing the verification test i.e. the probability that the state $\rho^R$ she sends to Alice passes all the $R$ runs of the SWAP test.
Even though the test runs are independent, if a generalised entangled state $\rho^R$ is sent by Eve, her success probability across the runs may no longer be the product of success probability of individual test runs. This implies that Eve's strategy might result in a higher success probability in some rounds based on the results of previous rounds. However, we show that since the $N$ distinct challenges being picked by Alice are all uniformly random, Eve does not gain anything by entangling the states across rounds corresponding to different challenges. To this end, we assume Eve can achieve optimal success probability by sending the state $ \bigotimes_{i=1}^{N}\rho_i^M$, where $\rho_i^M$ is a generalised state sent to $M$ runs of the SWAP test corresponding to the same challenge $\kc$.
Across these $j \in [M]$ runs corresponding to $\kc$, the state received by Alice is $\rho_{i,j} = \text{Tr}_{\{1\cdots M/j\}}(\rho_i^M)$, where $\rho_{i,j}$ is obtained by tracing out the M-1 instances $\{1,\cdots M/j \}$. Let $\rho_i^{\text{max}}$ be the Eve's response state corresponding to challenge $\kc$, with the highest fidelity with the correct response, i.e.
\begin{equation}
\small
   F(\rho_i^{\text{max}}, \kr) = \sqrt{\bra{\phi^r_i}\rho_i^{\text{max}}\kr} \geqslant \sqrt{\bra{\phi^r_i}\rho_{i,j}\kr \hspace{2mm}} \forall{j}\in M
    \label{eq:maxrho}
\end{equation}
Since the SWAP test success probability is directly proportional to the fidelity between the two input states, this implies that Eve can maximise her success probability by sending $M$ unentangled states $\rho_i^{\text{max}}$ to Alice instead of the generalised state $\rho_{i}^M$. The above Equation~\ref{eq:maxrho} can be used to bound Eve's success probability in passing Alice's verification test,
\begin{equation}
    \begin{split}
        \text{Pr}[\text{Ver accept}_{\text{Eve}}] &=  \text{Pr}\big[\prod_{i=1}^{R}\text{(SWAP}(\rho_i, \kr) = 1)\big] = \prod_{i=1}^{N}\text{Pr}\big[\prod_{j=1}^{M}\text{(SWAP}(\rho_{i,j}, \kr) = 1)\big]\\
        &\leqslant \prod_{i=1}^{N}\prod_{j=1}^{M}\text{Pr}\big[\text{SWAP}(\rho_i^{\text{max}}, \kr) = 1\big]\\
        &\leqslant \prod_{i=1}^{N}\Big(\frac{1}{2} + \frac{1}{2}F_i^2 \Big)^{M} = \epsilon
    \end{split}
    \label{eq:swapEve}
\end{equation}

where $\rho_i = \text{Tr}_{\{1\cdots R/i\}}(\rho^R)$, and $F_i =  F(\rho_i^{\text{max}}, \kr)$.
Now using the property that the qPUF device exhibits selective unforgeability against any QPT adversary Eve \cite{arapinis2021quantum}, we bound her success probability using the following theorem. 
\begin{theorem}[\textbf{SWAP Soundness}]\label{th:swap-sound} Let qPUF be a selectively unforgeable unitary PUF over $\HilD$ as defined in \cite{arapinis2021quantum}. The success probability of Eve to pass the {\normalfont SWAP}-test based verification of the {\cps} protocol is at most $\epsilon$, given that there are $N$ different CRPs, each with $M$ copies. The $\epsilon$ is bounded as follows:
{\normalfont
     \begin{equation}
     \text{Pr}[\text{Ver accept}_{\text{Eve}}] \leqslant \epsilon \approx \mathcal{O}(\frac{1}{2^{NM}})  \end{equation}
     }
\end{theorem}

\begin{proof}
From Eq~\ref{eq:swapEve}, we see that the optimal strategy of Eve is to produce the response states $\rho_i^{\text{max}}$ which maximises the fidelity $F_i$ for each CRP $(\kc, \krm)$. Arapinis et al. \cite{arapinis2021quantum} provided an upper bound on the fidelity when Eve has polynomial access to the qPUF. This property also referred to as the selective unforgeability property of qPUF (Appendix~\ref{sec:sel-unf}), states that the fidelity-square $F_i^2$ is bounded as,
\begin{equation}
    \text{Pr}[F_i^2 \geqslant \delta] \leqslant \frac{d + 1}{D} 
    \label{eq:selunfor}
\end{equation}
for any $\delta > 0$. Here $d = poly(\lambda) = poly\log(D)$ is the dimension of subspace that Eve has learnt from $\HilD$. For $D = 2^d$, this implies that the maximum fidelity state that Eve can create on average is non-orthogonal to the valid response state $\kr$ with a negligible probability $\approx \mathcal{O}(2^{-d})$. Hence $F_i^2 = \delta \rightarrow 0$ with overwhelming probability. This bound holds   true for all distinct CRPs labelled by $i \in [N]$.

Thus from Eq~\ref{eq:swapEve} and \ref{eq:selunfor}, the probability that Eve passes Alice's SWAP based verification test is,
\begin{equation}
\begin{split}
    \text{Pr}[\text{Ver accept}_{\text{Eve}}] &\leqslant \prod_{i=1}^{N}\Big(\frac{1}{2} + \frac{1}{2}F_i^2 \Big)^{M}\\ 
    &\leqslant  \prod_{i=1}^{N}\Big(\frac{1}{2} + \frac{1}{2}\delta \Big)^{M} \\
    &\approx \mathcal{O}(\frac{1}{2^{NM}}) = \negl(\lambda) 
\end{split}
\end{equation}

Note that here we also take into account the adaptive strategy of the adversary. That is even by assuming the previous rounds are added as extra states to Eve's learning phase, the dimension of the subspace $d$ will remain polynomial in $\lambda$. This completes the proof.
\end{proof}

The bound indicated above shows that one can achieve an exponentially secure qPUF-based identification using SWAP test based verification protocol with just a single challenge state i.e. $N = 1$ and repeated for $M$ instances. However, non-ideal cases would make identification with different challenge states necessary. Hence we provide a general recipe involving multiple distinct challenges each running for multiple instances. Our protocol requires $R = N \times M$ number of rounds and uses $T = 2R$ number of communicated states. 

\subsection{Verification with GSWAP test} \label{sec:gswapver}

The second proposal for Alice's \texttt{qVer} algorithm is the GSWAP test (Appendix~\ref{sec:test} Equation~\ref{eq:gswap}) and the identification protocol using this test is called {\cpg}. Its single run requires one copy of the received state and $M$ copies of Alice's response state and produces a binary outcome to determine the equality between two states with a polynomial one-sided error i.e. $\propto 1/M$. To boost the security to exponentially low error with a polynomial number of copies, Alice first runs the challenge phase with $R = N \subset K$ distinct challenge states, then uses the GSAWP test as \texttt{qVer} algorithm to test the equality. To this end, she consumes $N$ received response states and $N\times M$ numbers of valid response states in her database.
In the next two theorems, we show that GSWAP based test algorithm provides us with the desired completeness and soundness properties required in the protocol.

\begin{theorem}[\textbf{GSWAP Completeness}]\label{th:gswap-comp} In absence of Eve, the probability that Bob's response state generated from the valid qPUF $\kp = qPUF(\kc)$ passes all the $R = N$ test runs is,
{\normalfont 
\begin{equation}
\small \text{Pr}[\text{Ver accept}_{\text{H}}] =  \text{Pr}\big[\prod_{i=1}^{N}\text{(GSWAP}(\kp, \krm) = 1)\big] = 1    
\end{equation}
 }
\end{theorem}

\begin{proof}
When Alice receives Bob's response $\kp$ which is generated from the valid qPUF device for all the $i \in [R]$ copies of the challenge state, then $\kp = \kr$. This implies that $F(\kp,\kr) = 1$ for all $i \in [R]$. From Eq~\ref{eq:gswap}, we see that,
\begin{equation}
\text{Pr}\big[(\text{GSWAP}(\kp, \krm) = 1] = \frac{1}{M+1} + \frac{M}{M+1}F^2(\kp,\kr) = 1, \hspace{2mm} \forall i \in [N]  
\end{equation}
Since in the honest setting, the states received from Bob over $R$ rounds are all valid qPUF pure states which are unentangled to each other, hence the GSWAP tests for all the $R$ rounds are independent tests. This implies that,
\begin{equation}
        \text{Pr}[\text{Ver accept}_{\text{H}}] =  \text{Pr}\big[\prod_{i=1}^{N}\text{(GSWAP}(\kp, \kr) = 1)\big] = \prod_{i=1}^{N}\text{Pr}\big[\text{GSWAP}(\kp, \kr) = 1\big] = 1
\end{equation}
This completes the proof.
\end{proof}
To characterise the soundness, we bound Eve's success probability in simultaneously passing the $N$ runs of GSWAP test when she sends the generalised entangled state $\rho^N$ to Alice. Similar to the argument provided in the SWAP test soundness, Eve does not gain anything by entangling the states across different test runs. Thus Eve's probability in passing the verification test by sending the state $ \bigotimes_{i=1}^{N}\rho_i$ is the same as that for a generalised state $\rho^N$, where $\rho_i$ is the state sent to the instance of GSWAP test corresponding to the same challenge $\kc$. As a result, Eve's optimal success probability can be expressed as a product of individual GSWAP instance success probability,
\begin{equation}
    \begin{split}
        \text{Pr}[\text{Ver accept}_{\text{Eve}}] &=  \text{Pr}\big[\prod_{i=1}^{N}\text{(GSWAP}(\rho_i, \krm) = 1)\big] = \prod_{i=1}^{N}\text{Pr}\big[\text{GSWAP}(\rho_i, \krm) = 1\big]\\
        &\leqslant \prod_{i=1}^{N}\Big(\frac{1}{M+1} + \frac{M}{M+1}F_i^2 \Big) = \epsilon
    \end{split}
    \label{eq:gswapEve}
\end{equation}

where $F_i =  F(\rho_i, \kr)$ is the fidelity between Eve's state and the valid qPUF response state for the $i$-th round.

\begin{theorem}[\textbf{GSWAP Soundness}]\label{th:gswap-sound} Let qPUF be a selectively unforgeable unitary PUF over $\HilD$ as defined in \cite{arapinis2021quantum}. The success probability of Eve to pass the {\normalfont GSWAP}-test based verification of the {\cpg} protocol is at most $\epsilon$, given that there are $N$ different CRPs, each with $M$ copies. The $\epsilon$ is bounded as follows:
{\normalfont 
\begin{equation}
    \text{Pr}[\text{Ver accept}_{\text{Eve}}] \leqslant \epsilon \approx \mathcal{O}\big(\frac{1}{(M+1)^{N}}\big)
\end{equation}
}
\end{theorem}

\begin{proof}
From Eq~\ref{eq:gswapEve}, we see that the optimal strategy of Eve is to produce the response states $\rho_i$ which maximises the fidelity $F_i$ for each CRP $(\kc, \krm)$. We utilise the same selective unforgeability result (Appendix~\ref{sec:sel-unf}) to bound the fidelity-square $F_i^2$ with which Eve can produce the states $\rho_i$,
\begin{equation}
    \text{Pr}[F_i^2 \geqslant \delta] \leqslant \frac{d + 1}{D} 
    \label{eq:gselunfor}
\end{equation}
for any $\delta > 0$. Here $d = poly(\lambda) = poly\log(D)$ is the dimension of subspace that Eve has learnt from $\HilD$. For $D = 2^d$, this implies that the maximum fidelity state that Eve can create on average is non-orthogonal to the valid response state $\kr$ with a negligible probability $\approx \mathcal{O}(2^{-d})$. Hence $F_i^2 = \delta \rightarrow 0$ with overwhelming probability. This bound holds   true for all distinct CRPs labelled by $i \in [N]$.

Thus from Eq~\ref{eq:gswapEve} and \ref{eq:gselunfor}, the probability that Eve passes Alice's SWAP based verification test is,
\begin{equation}
\begin{split}
    \text{Pr}[\text{Ver accept}_{\text{Eve}}] &\leqslant \prod_{i=1}^{N}\Big(\frac{1}{M+1} + \frac{M}{M+1}F_i^2 \Big)\\ 
    &\leqslant  \prod_{i=1}^{N}\Big(\frac{1}{M+1} + \frac{M}{M+1}\delta \Big) \approx \mathcal{O}\big(\frac{1}{(M+1)^{N}}\big) = \negl(\lambda) 
\end{split}
\end{equation}

Note that here we have also taken into account the adaptive strategy of Eve since our security is analysed for the most general attack strategy. This completes the proof.
\end{proof}

The recent bound shows that to achieve an exponentially secure qPUF based identification using GSWAP based verification protocol with only a polynomial sized register $S$, the protocol needs to be repeated for multiple $N$ instances. Our protocol requires $R = N$ number of communication rounds and uses $T = 2R$ number of communicated states. 

\section{qPUF identification protocol with low-resource verifier}
\label{sec:qp}
Our second protocol enables a weak verifier to identify a quantum server prover in the network. We achieve this by delegating the equality testing to the prover thus effectively removing the quantum computational requirement on the verifier. While this might look like it could facilitate a malicious Eve to fool the weak verifier easily, we demonstrate due to the unforgeability of qPUF that the security is not affected. Before describing the details, we list the salient features of our protocol,
\begin{itemize}
    \item The protocol requires the prover to have some quantum memory and high-recourse computing ability, whereas the verifier is just required to have quantum memory and no computing ability resource during the identification and verification phase\footnote{The state preparation phase happens in the setup phase of the protocol and it is a common property of all qPUF-based protocols. Hence here we are mostly interested in the computing ability in the verification phase, which is the key difference between such protocols due to the fact that verifying quantum states is a challenging task.} (restricted to QPT memory and computation).
    \item The protocol requires a 1-way quantum communication link directed from the verifier to the prover. The prover to the verifier directed link is a classical channel.
    \item The protocol has a classical verification phase i.e. the prover locally performs the verification test and sends the classical information to the verifier.
    \item The protocol provides perfect completeness and an exponentially-high security guarantee against any adversary with QPT resources. 
\end{itemize}

\begin{figure}[ht!]
\includegraphics[scale=0.5]{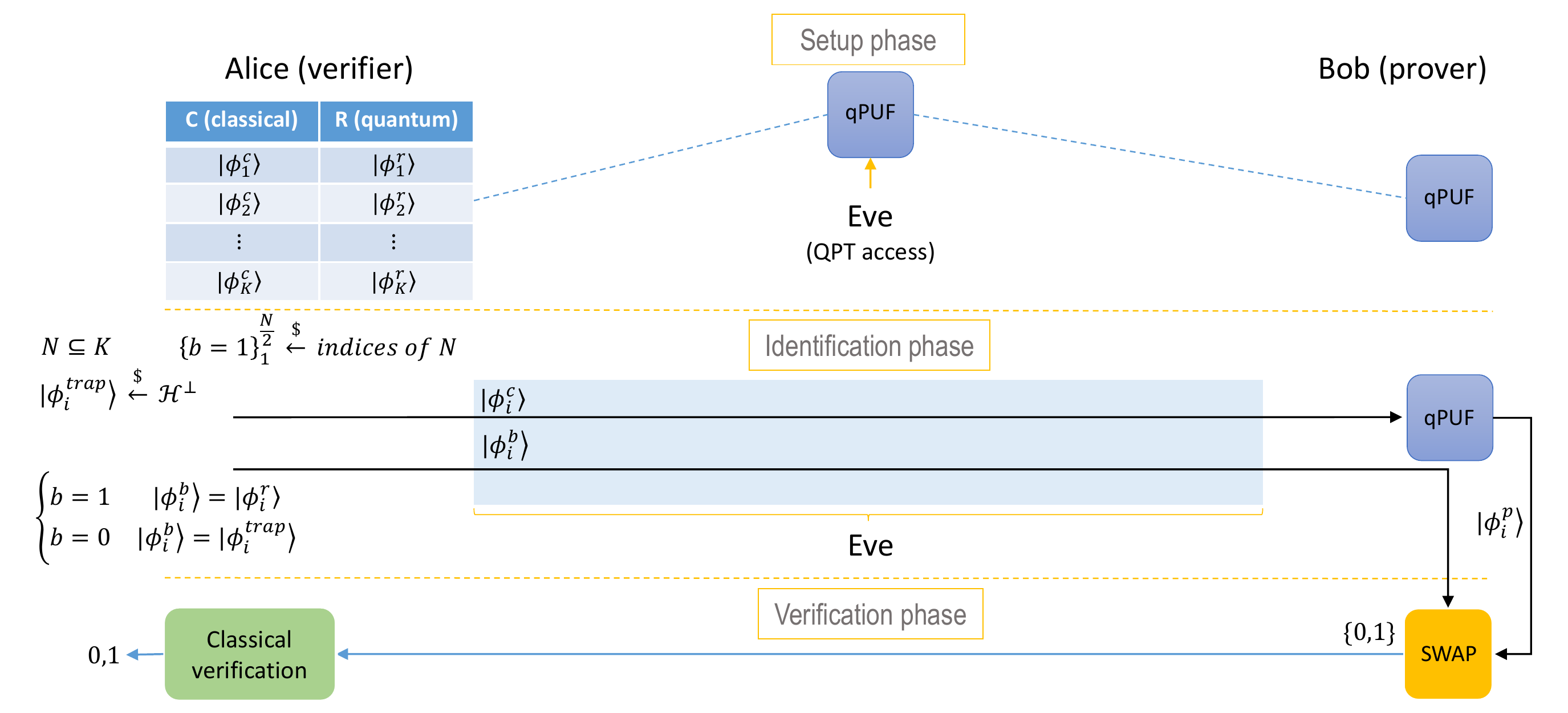}
\centering
\caption{qPUF-based identification protocol with low-resource verification between Alice (verifier) and Bob (prover) (\qp). The protocol is divided into three sequential phases, \emph{setup phase}, \emph{identification phase} and \emph{verification phase}. In the identification phase, Alice randomly picks a subset $N \subseteq K$ of challenges which are sent to Bob. Further, to correctly identify Bob, she employs a trap based scheme where she sends either the correct response state of the challenges or the trap states which are states orthogonal to the valid response states. Bob performs the SWAP-test based verification algorithm and sends the classical bits back to Alice. Alice performs a check on the received bits and outputs a classical bit `1' if Bob's device is correctly identified. Otherwise, she outputs `0'}  
\label{fig:qid-p2}
\end{figure}
\subsection{Protocol description}

This protocol is run between Alice, the verifier, and Bob, the prover in three sequential phases,
\begin{enumerate}
    \item \emph{Setup phase}:
            \begin{enumerate}
                \item Alice has the qPUF device. 
                \item Alice randomly picks $K \in \mathcal{O}(\text{poly} \log D)$ classical strings $\phi_i \in \{0,1\}^{\log D}$.
                \item Alice selects and applies a Haar-random state generator operation denoted by the channel $\E$ to locally create the corresponding quantum states in $\HilD$: $\phi_i \overset{\E}{\rightarrow} \kc,\hspace{2mm} \forall i \in [K]$.
                \item Alice queries the qPUF individually with each quantum challenge $\kc$ to obtain the response state $\kr$.
                \item Alice creates states $\ket{\phi_i^{\perp}}$ orthogonal to $\kc$ and queries the qPUF device with them to obtain the trap states labelled as $\ket{\phi_i^{\text{trap}}}$. The unitary property of qPUF device ensures that $\langle \phi_i^{\text{trap}}|\phi_i^r\rangle = 0$. 
                \item She creates a local database $S \equiv \{\kc,\{ \kr, \ket{\phi_i^{\text{trap}}}\}\}$ for all $i \in [K]$. Thus the $S$ registers stores the challenge state $\kp$ and the corresponding valid response state and the trap state which is orthogonal to the response state.
                \item Alice publicly transfers the qPUF to Bob.
            \end{enumerate}
            The transition is non-secure and Eve is allowed $\mathcal{O}(\text{poly} \log D)$ query access to the qPUF to build her database. 
    \item \emph{Identification phase}:
            \begin{enumerate}
                \item Alice randomly selects a subset $N \subseteq K$ different challenges $\kc$ and sends them over a public channel to Bob.
                \item She randomly selects $N/2$ positions, marks them $b = 1$ and sends the valid response states $\ket{\phi_i^1} = \kr$ to Bob. On the remaining $N/2$ positions, marked as $b = 0$, she sends the trap states $\ket{\phi_i^0} = \ket{\phi_i^{\text{trap}}}$.
            \end{enumerate}
    \item \emph{Verification phase}:
            \begin{enumerate}
                \item Bob queries the qPUF device with the challenge states received from Alice to generate the response states $\kp$ for all $i \in [N]$. 
                \item He performs a quantum equality test algorithm by performing a SWAP test between $\kp$ and the response state $\ket{\phi_i^b}$ received from Alice. This algorithm is repeated for all the $N$ distinct challenges.   
                \item Bob labels the outcome of $N$ instances of the SWAP test algorithm by $s_i \in \{0,1\}$ and sends them over a classical channel to Alice.
                \item Alice runs a classical verification algorithm \texttt{cVer($s_1,...,s_N$)} and outputs `1' implying that Bob's qPUF device has been successfully identified. She outputs `0' otherwise. 
            \end{enumerate}
\end{enumerate}

Figure~\ref{fig:qid-p2} describes the q-PUF based identification protocol with low-resource verification denoted as {\qp}. 
For the {\qp} protocol, completeness is the probability that Alice's verification algorithm $\texttt{cVer}$ returns an outcome `1' in absence of Eve. Ideally we require completeness to differ negligibly from 1,
\begin{equation}
     \text{Pr}[\text{Ver accept}_{\text{H}}] =  \text{Pr}\big[\texttt{cVer}(S_N) = 1\big] = 1 - \negl(\lambda)
     \label{eq:complete2}
\end{equation}
where $\lambda$ is the security parameter. 

Soundness of the protocol is the probability that \texttt{cVer} returns an outcome `1' in presence of Eve. For security, we require the soundness to be negligible in $\lambda$,
\begin{equation}
    \text{Pr}[\text{Ver accept}_{\text{Eve}}] =  \text{Pr}\big[\texttt{cVer}(S_N) = 1\big] =  \negl(\lambda)
    \label{eq:gensound2}
\end{equation}
We investigate the security of our protocol when Bob uses the SWAP test and Alice uses the classical verification algorithm $\texttt{cVer}$. We remark that Bob can alternatively use GSWAP testing to generate the outcomes, however, this would require Alice to send multiple copies of the same challenge state to Bob, thus incurring higher resources on Alice's side.

\subsection{\texttt{cVer} algorithm}
The main ingredient of verification is the \texttt{cVer} classical test algorithm employed by Alice to certify whether Bob's device has been identified.  As described in Algorithm~\ref{alg:cver}, 
\texttt{cVer} receives an $N$-bit binary string $S_N$ as input. The algorithm is divided into two tests. \texttt{test1} first checks whether in the $N/2$ positions marked as $b = 1$, i.e. the positions where Alice had sent a valid qPUF response state to Bob, if the corresponding bits in $S_N$ are all 0.

If this test succeeds, then the algorithm proceeds to \texttt{test2} which is a test on the positions where Alice had sent the trap states to Bob. If on these positions, the expected number of bits in $S_N$ which are 0 lie between $\{\kappa\frac{N}{2} - \delta_{er}, \kappa\frac{N}{2} + \delta_{er}\}$, then \texttt{cVer} algorithm outputs `1' indicating that the device has been identified. Here $\kappa\frac{N}{2}$ is the expected number of bits in $b=1$ positions with outcome `0' that Bob would obtain after the Equality test algorithm measurement, in absence of any adversary Eve.In our case when Bob uses SWAP test, $\kappa = 0.5$. Here, $\delta_{er}$ accounts for the statistical error in the measurement.

\begin{algorithm}[ht!]
\SetAlgoLined
\textbf{Description:} Let $S_N = \{0,1\}^N$ be the input $N$-bit string. Let $P=\{i_k\}^{N/2}_{k=1}$ be the set of indices showing the rounds of the protocol where $b=1$. Algorithm consists of two tests, \texttt{test1} and \texttt{test2} as follows:\\
   \texttt{test1:}\\
   \ForAll{$i$ in P}{
    \If{$s_{i} = 0$}{
        $count\gets count+1$\;
    }}
    \eIf{$count = \frac{N}{2}$}{
        \Return 1\;
    }{\Return 0\;}

\hrulefill\\
\texttt{test2:}\\
    \eIf{\texttt{test1} = 0}{
        \Return 0\;
    }
    {
    \ForAll{$i$ not in P}{
    \If{$s_{i} = 1$}{
        $count\gets count+1$\;
    }}
    \eIf{$\lvert count - \delta\frac{N}{2} \rvert \leqslant \delta_{er}$}{
        \Return 1\;
    }{\Return 0\;}
    }
 \caption{\texttt{cVer} algorithm}
\end{algorithm}\label{alg:cver}

\subsection{Verification using SWAP test and \texttt{cVer} algorithm}

Here we explicitly describe and calculate the completeness and soundness probabilities of the $\qp$ protocol which employs the verification algorithm involving Bob's SWAP test, followed by Alice's \texttt{cVer} algorithm. This allows Alice to efficiently identify the valid qPUF device even though the SWAP test algorithm has been delegated to Bob. A single instance of Bob's SWAP test requires a single copy of the response state received from Alice (either the valid qPUF response state or the trap state) and the response state that Bob generates by querying Alice's challenge state in his qPUF device. To obtain a desired low enough error rate in the verification algorithm, the SWAP test is performed on $N$ distinct instances of the received response state and response state generated by Bob by querying distinct challenges states. The responses of the SWAP test instances are classical bits. Thus the $N$ bit binary classical outcome string is sent to Alice who employs the algorithm \texttt{cVer} described in Algorithm~\ref{alg:cver}.
An identification protocol performed using $N$ distinct challenge states consumes a combined total of $2N$ copies of the received state and the response state generated by Alice. In the next two sections, we show that SWAP based test algorithm provides us with the desired completeness and soundness properties required in the protocol.

\begin{theorem}[\textbf{\texttt{cVer} Completeness}]\label{th:cver-comp} In absence of Eve, the probability that the $N$-bit string $S_N = \{s_1,...,s_N\}$ sent by Bob, passes the {\normalfont \texttt{cVer}($S_N$)} algorithm is,
{\normalfont 
\begin{equation}
\small  \text{Pr}[\text{Ver accept}_{\text{H}}] =  \text{Pr}\big[\texttt{cVer}(S_N) = 1\big] = 1 - 2e^{-N/4}     
\end{equation}
}
\end{theorem}

\begin{proof}
To prove this theorem, we separately analyse the $N/2$ positions where Alice sends the valid qPUF response state to Bob (marked as $b =1$), and the remaining positions where she sends the trap state (marked as $b = 0$),
\begin{enumerate}
    \item $b = 1$ positions: When Bob prepares the response state $\kp$ by querying her qPUF device with Alice's challenge state $\kc$, then Bob's generated response state is equal to Alice's response state sent to Bob, i.e. $\kr = \kp$. This implies that $F(\kp,\kr) = 1$ for all $i \in [N]$ marked $b=1$. From Eq~\ref{eq:swapaccept}, we see that,
    \begin{equation}
    \text{Pr}\big[\text{SWAP}(\kp, \kr) = 1] = \frac{1}{2} + \frac{1}{2}F(\kp,\kr)^2 = 1, \hspace{2mm}   
    \end{equation}
    From section~\ref{sec:swap}, we see that $[\text{SWAP}(\kp, \kr) = 1]$ corresponds to the classical outcome 0. This implies that $s_i = 0$ for all $i \in [N]$ marked $b=1$ with certainty. Thus when Alice employs the \texttt{cVer} algorithm, Bob always achieves a $count = N/2$ in the \texttt{test1} and thus passes it with certainty,
    \begin{equation}
        \text{Pr}[\texttt{test1 pass}] = 1
    \end{equation}    
    \item $b = 0$ positions: These positions correspond to Alice sending the trap states $\ket{\phi_i^{\text{trap}}}$ to Bob such that Bob's generated response state $\kp$ is orthogonal to the trap state. In other words, $F(\kp, \ket{\phi_i^{\text{trap}}}) = 0$ for all $i \in [N]$ marked $b=0$. This implies that,
    \begin{equation}
    \text{Pr}\big[\text{SWAP}(\kp, \ket{\phi_i^{\text{trap}}}) = 1] = \frac{1}{2} + \frac{1}{2}F(\kp,\ket{\phi_i^{\text{trap}}})^2 = \frac{1}{2}, \hspace{2mm}
    \end{equation}
    Thus, half of the $N/2$ positions would produce the classical outcome 1 on average. When Alice employs \texttt{test2} of the \texttt{cVer} algorithm, $\mathbb{E}[count] = N/4$. Using the Chernoff-Hoeffding inequality \cite{upfal2005probability}, for any constant $\delta_{er} > 0$,
    \begin{equation}
        \text{Pr}[\texttt{test2 pass}] = \text{Pr}\Big[\Big\lvert count - \frac{N}{4}\Big\rvert \leqslant \delta_{er}\Big] \geqslant 1 - 2e^{-N\delta_{er}^2}
    \end{equation}
\end{enumerate}
From the above results and using the fact that $\delta_{er} = 0.5$ for SWAP test based algorithm,
\begin{equation}
    \begin{split}
        \text{Pr}[\text{Ver accept}_{\text{H}}] &=  \text{Pr}\big[\texttt{cVer}(s_1,\cdots, s_N) = 1\big]\\
        &= \text{Pr}\big[\texttt{test1 pass} \wedge \texttt{test2 pass}\big]\\
        &= \text{Pr}[\texttt{test1 pass}]\cdot \text{Pr}[\texttt{test2 pass}]\\
        &\geqslant 1 - 2e^{-N/4}
    \end{split}
\end{equation}
This completes the proof.
\end{proof}

The next section details the soundness proof of the {\qp} protocol.

\section{{\qp} protocol soundness}\label{sec:qp-soundness}

To characterise the soundness, we bound Eve's success probability in passing the \texttt{cVer} test. Since the verification test is reduced to a classical test, we consider the soundness in the presence of two types of Eve. The first is a \emph{classical Eve} who does not process any quantum resources. The second is a \emph{quantum Eve}, who possess QPT memory and computing capability. We separately analyse the security against both types of Eve and prove that \emph{quantum Eve} gains only exponentially small advantage compared to the \emph{classical Eve}, thus reducing the security to analysing only the classical adversary. We show that since the verification test is classical, the only way for a \emph{quantum Eve} to succeed better than a \emph{classical Eve} is to succeed at guessing the trap positions better than a random guess of \emph{classical Eve}. We utilise the unforgeability property of qPUF to prove that a \emph{quantum Eve} can have an only negligible advantage in guessing the trap positions compared to a \emph{classical Eve}, thus enabling the reduction. 

\subsection{Security against classical adversary}

\begin{theorem}[\textbf{Soundness against classical Eve}]\label{th:cv-clattack} The probability that any classical Eve produces an $N$-bit string $S_N = \{s_1,...,s_N\}$ which passes the {\normalfont \texttt{cVer}} algorithm is,
{\normalfont 
\begin{equation}
 \text{Pr}[\text{Ver accept}_{\text{Eve}}] =  \text{Pr}\big[\texttt{cVer}(S_N) = 1\big] \leqslant \mathcal{O}(2^{-N})      
\end{equation}
}
\end{theorem}

\begin{proof}
Here we provide a proof sketch. The detailed proof is provided in the Appendix~\ref{ap:lrv-sound}. We remark that any classical Eve's strategy to produce a valid $N$-bit string $S_N$ can be divided into two categories,
\begin{enumerate}
    \item \textbf{Independent guessing strategy:} Eve independently guesses each bit of the string $S_N$ that would pass \texttt{cVer} algorithm. 

    \item \textbf{Global strategy:} Eve outputs a string $S_N$ using the global properties of the \texttt{cVer} such that $S_N$ passes the verification test with maximum probability. In contrast to the previous strategy, the probability to output each bit $s_i$ is not necessarily independent.  
\end{enumerate}
We calculate the success probability of Eve in both cases and by optimising over both strategies, we obtain a higher success probability when Eve employs global strategy. 
The two strategies, however, converge in the limit of large $N$.

Under the independent strategy, when Eve guesses each bit with a probability $\{\alpha, 1 - \alpha\}$, one obtains the probability with which the resulting $n$-bit string of this strategy passes the \texttt{cVer} test. Maximising Eve's passing probability over all $\alpha \in [0,1]$, we obtain an optimal $\alpha=3/4$. The resulting Eve's accept probability is,
\begin{equation}
\small
     \text{Pr}[\text{Ver accept}_{\text{Eve,Ind}}] = (2\delta_{er} + 1) \frac{3^{\frac{3N}{4}}}{2^{2N}}\times {N/2\choose N/4} \approx \mathcal{O}(2^{-N})
\end{equation}
The second category is the global strategy where Eve optimises over all the strategies of guessing the $N$ bit string which passes \texttt{cVer} with maximum probability. In order to find the optimal global strategy, we extract out the essential properties leveraged by Eve to pass the test. We note that since the trap response positions are chosen uniformly at random by Alice, hence Eve does not have any information on the index set $P$ of Algorithm~\ref{alg:cver}. Eve, however, knows the statistics of 0's and 1's in $S_N$ to pass \texttt{cVer}. For instance, a string must have a minimum of $3N/4 - \delta_{er}$ bits which are $0$, otherwise, the string necessarily fails the \texttt{test1} or \texttt{test2} or both. Based on statistics knowledge, any global strategy for Eve should consist of optimising the number of 0's and 1's to pass \texttt{test1} and \texttt{test2}. We show that an optimal global strategy $\E_{gop}$ is the one that outputs a string $S_{gop}$ with the number of `1' bits $c_1 \in m_{valid} = \{\frac{N}{4} - \delta_{er}, \dots,  \frac{N}{4} + \delta_{er}\}$. To prove optimality of this strategy we show that any other strategy necessarily fails the \texttt{cVer} test. To calculate the success probability of Eve under $\E_{gop}$ strategy, we define the event space- the set of potentially valid strings with $c_1 \in m_{valid}$ that Eve needs to choose from to maximise her accept probability. From this, we calculate the domain size of the subset of strings that pass \texttt{test1}. If this test is passed, then \texttt{test2} is automatically passed since Eve chooses from the event space.  The resulting Eve's success probability is,
\begin{equation}
\small
 \text{Pr}[\text{Ver accept}_{\text{Eve,G}}] =
(2\delta_{er} + 1)\times \frac{(\frac{N}{2})!(\frac{3N}{4})!}{N!(\frac{N}{4})!}
\approx \mathcal{O}(2^{-N})
\end{equation}
The exponential security in $N$ comes from the fact that the event space is exponentially large compared to the subset. This is a consequence of the random hiding of traps by Alice. 

We compare the two attack strategies of Eve to find the optimal classical attack. For this comparison, we fix the accepted tolerance value $\delta_{er}=1$ although the same result holds for $\delta_{er} \neq 0$. Figure~\ref{fig:Figure3} shows the acceptance probabilities of Eve in the independent guessing strategy and global strategy as a decreasing function of the string length $N$. It is clear that the global strategy performs better than the independent strategy, although they both converge in the limit of large $N$.
\begin{figure}[h!]
\includegraphics[scale=0.6]{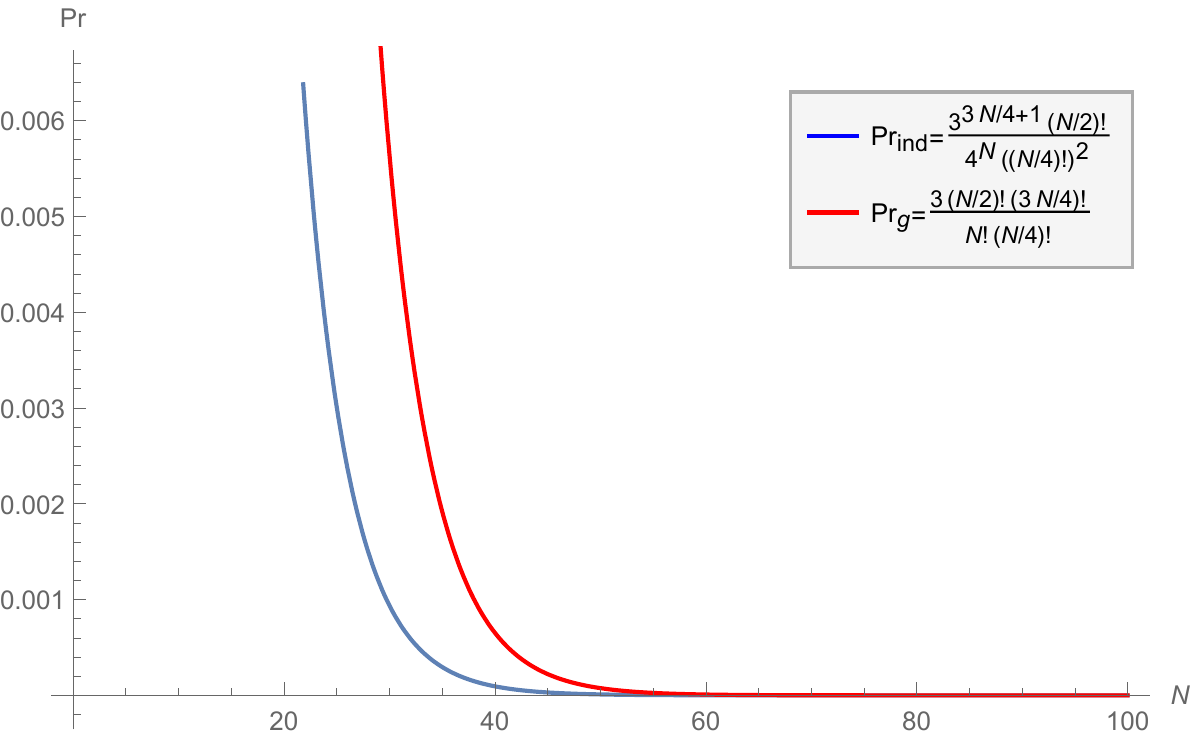}
\centering
\caption{Comparison of the acceptance probabilities of Eve in the independent guessing strategy (in blue) and global strategy (in red) as a decreasing function of the string length $N$ for the tolerance value $\delta_{er}=1$}
\label{fig:Figure3}
\end{figure}
\end{proof}

\subsection{Security against quantum adversary}

We now investigate the soundness property of the protocol against QPT Eve by modelling Eve's strategy with a completely positive trace preserving (CPTP) map that takes as input the target challenge $\kc$, the unknown state $\kb$, and ancilla qubits and outputs the classical bits which are sent to Alice for verification. This map utilises the database information created by Eve during the qPUF transition. A QPT Eve's strategy can be divided into two categories,
\begin{enumerate}
    \item \textbf{Collective attack strategy}: Eve applies an independent CPTP map on each of the $N$ rounds.
    \item \textbf{Coherent attack strategy}: Eve applies a CPTP map on the combined N distinct challenge and their corresponding response states that Alice sends to Bob.
\end{enumerate}
A collective strategy is a special case of Eve's coherent strategy. However, we show that independence in choosing the trap states by Alice reduces the coherent strategy to the collective strategy by Eve. We analyse the collective security first and then give a reduction of the coherent strategy to the collective strategy.

\subsubsection{Collective strategy:} Under this strategy, Eve optimises over all the CPTP maps that inputs Alice's states $\kc$ and $\kb$ and outputs a single bit $s_i$ to maximise the acceptance probability. 
Figure~\ref{fig:Figure4} shows Eve performing a general collective strategy. 
\begin{figure}[h!]
\includegraphics[scale=0.5]{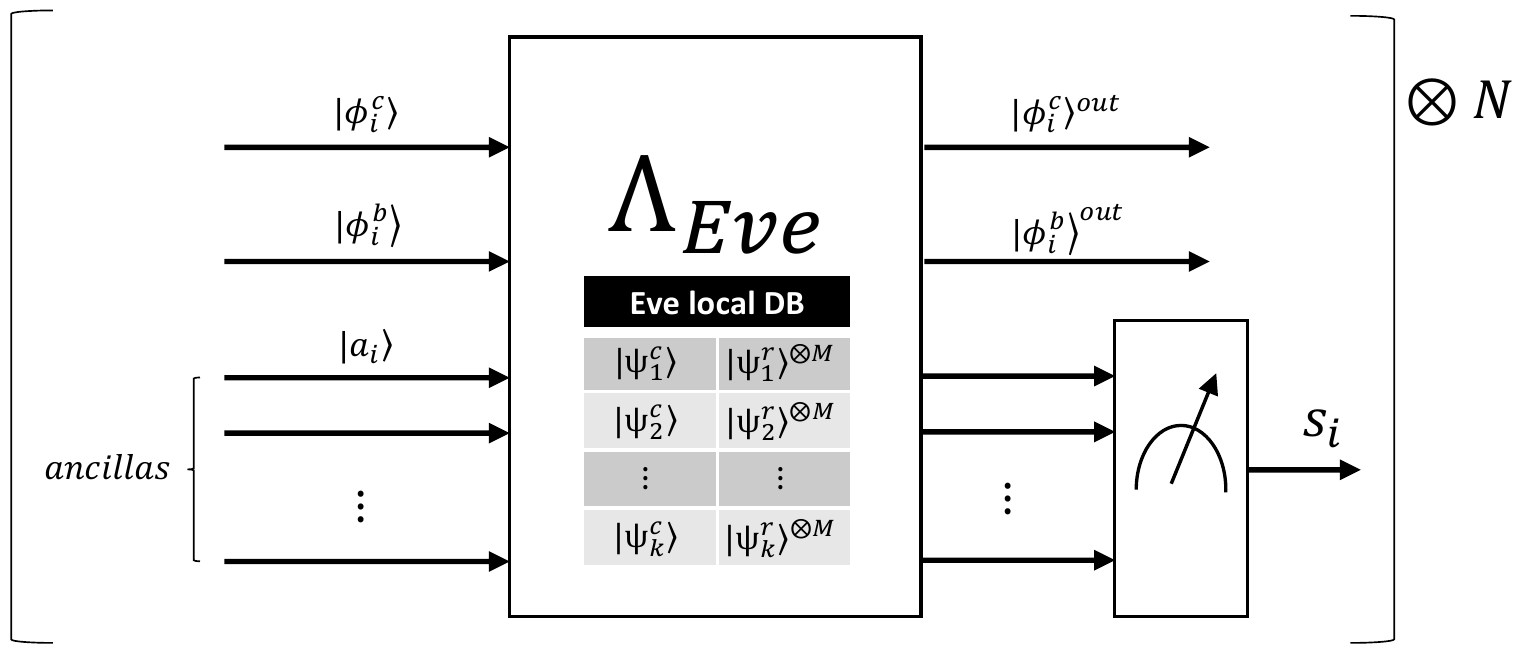}
\centering
\caption{Quantum collective attack strategy performed by Eve on {\qp} protocol by applying the same local-database-dependent CPTP map on each round of the challenge and response state $\kc$ and $\kb$ respectively. The output of the single instance of the map is a bit $S_i$.}
\label{fig:Figure4}
\end{figure}

We denote Eve's quantum map to be, 
\begin{equation}
    \Lambda_{Eve} \equiv \bigotimes_{i=1}^{N}\Lambda_i \quad.
\end{equation}
Contrary to the classical Eve who is unable to figure out the trap positions in any round with a probability higher than half, a QPT Eve, by leveraging her local database information, could be expected to do better than a random guess.
More formally, we say that the {\qp} protocol is secure against any QPT Eve who performs a CPTP map $\Lambda_i$ on the states $\kc, \kb$ for each $i \in [N]$, if the resulting success probability of correctly guessing the bit $b$ for each position differs negligibly in the security parameter from half.
%
\begin{theorem}[\textbf{Security against collective attack}] \label{th:colattack} The success probability of any QPT adversary in correctly guessing whether $\kb = \kr$ for each $i \in [N]$ differs negligibly from half,
\begin{equation}
    \text{Pr}[b \leftarrow \Lambda_i(\kc,\kb)] \leqslant \frac{1}{2} + \mathcal{O}(2^{- d}) \hspace{2mm} \forall i \in [N]
\end{equation}
where \normalfont{$d = \mathcal{O}(\text{poly} \log D)$} is the size of Eve's database and qPUF is in $\mathcal{H}^{D}$.
\end{theorem}
\begin{proof}
First, we use the symmetry of the problem to restrict ourselves to cases where $b=1$. We prove the theorem by contradiction i.e., suppose there exists an algorithm $W$ that wins the quantum security game for each index $i \in [N]$ with a probability non-negligibly better than a random guess. In other words, $W = 1$ if the index $b$ is correctly guessed, and $W = 0$ otherwise. Let $f(\lambda) \geqslant 0$ be a non-negligible function of the security parameter. 
The joint probabilities for all collective possible values of $b$ and $W$ can be written as,
\begin{equation}
\begin{split}
        & \text{Pr}[W=1, b=1] = \frac{1}{4} + f(\lambda)\quad \text{Pr}[W=1, b=0] = \frac{1}{4} - f(\lambda) \\
        & \text{Pr}[W=0, b=0] = \frac{1}{4} + f(\lambda)\quad \text{Pr}[W=0, b=1] = \frac{1}{4} - f(\lambda)
\end{split}
\end{equation}
where the joint probabilities are higher when $W$ correctly guesses $b$, and is lower otherwise. From this, we can define the following conditional probability of winning for cases where $b=1$ as follows:
\[Pr[W=1 | b=1] = \frac{Pr[W=1, b=1]}{Pr[b=1]} = \frac{1}{2} + f'(\lambda)\]
Where $f' = 2f$ is again a non-negligible function in the security parameter $\lambda$. This is the same probability of winning when $b = 0$ i.e. $\text{Pr}[W = 0|b = 0]$.

Now we show that the success probability of Eve in successfully guessing whether $\kb = \kr$ reduces to finding a CPTP map $\Lambda_i$ which performs an optimal quantum test to distinguish the response state $\kb$ with the reference state $\ke$. As Eve has no access to the actual response $\kr$, the reference state $\ke$ should be generated within the $\Lambda_i$ itself. Thus without loss of generality, any attack map $\Lambda_i$, consists of two parts. The first part uses a generator algorithm \emph{gen} to generate a reference state $\ke$, or more generally a mixed state $\rho_e$ by using the local database and the input challenge state $\kc$, and the second part performs a test algorithm $\T$ on $\kb$ and $\rho_e$,

\begin{equation}
    \Lambda_i \equiv \T(\kb, \rho_e \leftarrow gen(DB,\kc))
\end{equation}

where $DB$ is the local database of Eve generated in the  \emph{setup phase}. To further provide the capability to Eve, we assume that her test $\T$ is an optimal test equality test algorithm, also referred as ideal test algorithm in  Definition~\ref{def:ideal-test}, i.e. $\T = \Ti$. Note that $\Ti$ is the optimal test allowed by quantum mechanics where the probability of succeeding in the equality test is proportional to the square of the fidelity distance of the two states. Now we state the following contraposition:
Let us assume that there exists a winning algorithm $W$ running $\Lambda = \Ti(\kb, \rho_e)$ such that,
\begin{equation}
\text{Pr}[1 \leftarrow \Lambda(\kc,\kb)| b = 1] \leqslant \frac{1}{2} + \nonnegl(\lambda)   
\label{eq:nonnegl}
\end{equation}

From Definition~\ref{def:ideal-test}, we see that $\Ti$ outputs $1$ with probability $p = F(\kb, \rho_e)^2$. In other words,
\begin{equation}
     \text{Pr}[1 \leftarrow \Lambda(\kc,\kb)| b = 1] = \text{Pr}[1 \leftarrow \Ti] = F(\kb, \rho_e)^2 \leqslant \frac{1}{2} + \nonnegl(\lambda)
\end{equation}

This implies that if an algorithm $W$ exists for Eve, then she is able to generate the state $\rho_e$ with non-negligible fidelity with the valid qPUF response (for b=1), and similarly with trap states (for b = 0). And this would hold for all $i \in [N]$. 
But this contrasts with the selective unforgeability of the qPUF which states that the success probability of any QPT adversary having polynomial-size access to the qPUF is bounded as $\frac{d_e+1}{D}$ where $d_e = poly(\lambda) = \text{poly} \log(D)$ is the dimension of subspace that Eve has learnt from $\HilD$ \cite{arapinis2021quantum}. Thus such $\Lambda$ cannot exist even with the most efficient test $\Ti$. This proves the theorem.
\end{proof}

\subsubsection{Coherent Strategy:} The collective strategy is restricted to Eve applying individual unentangled maps in each round. A more generalised strategy, the coherent strategy, involves applying a CPTP map collectively on all the rounds thus potentially leveraging entanglement capabilities across rounds. Such a strategy takes as input the $N$ challenge states $\otimes_{i=1}^{N}\kc$, the $N$ response state $\otimes_{i=1}^{N}\kb$ and the ancilla qubits, and outputs a $N$ bit string $S_N$ which is sent to Alice for verification.  Figure~\ref{fig:Figure5} depicts this strategy.  Eve's objective is to produce the $S_N$ which maximises the \texttt{cVer} passing probability. We denote Eve's quantum map to be,
\begin{equation}
    \Lambda_{Eve} \equiv \Lambda^{N}
\end{equation}
\begin{figure}[h!]
\includegraphics[scale=0.55]{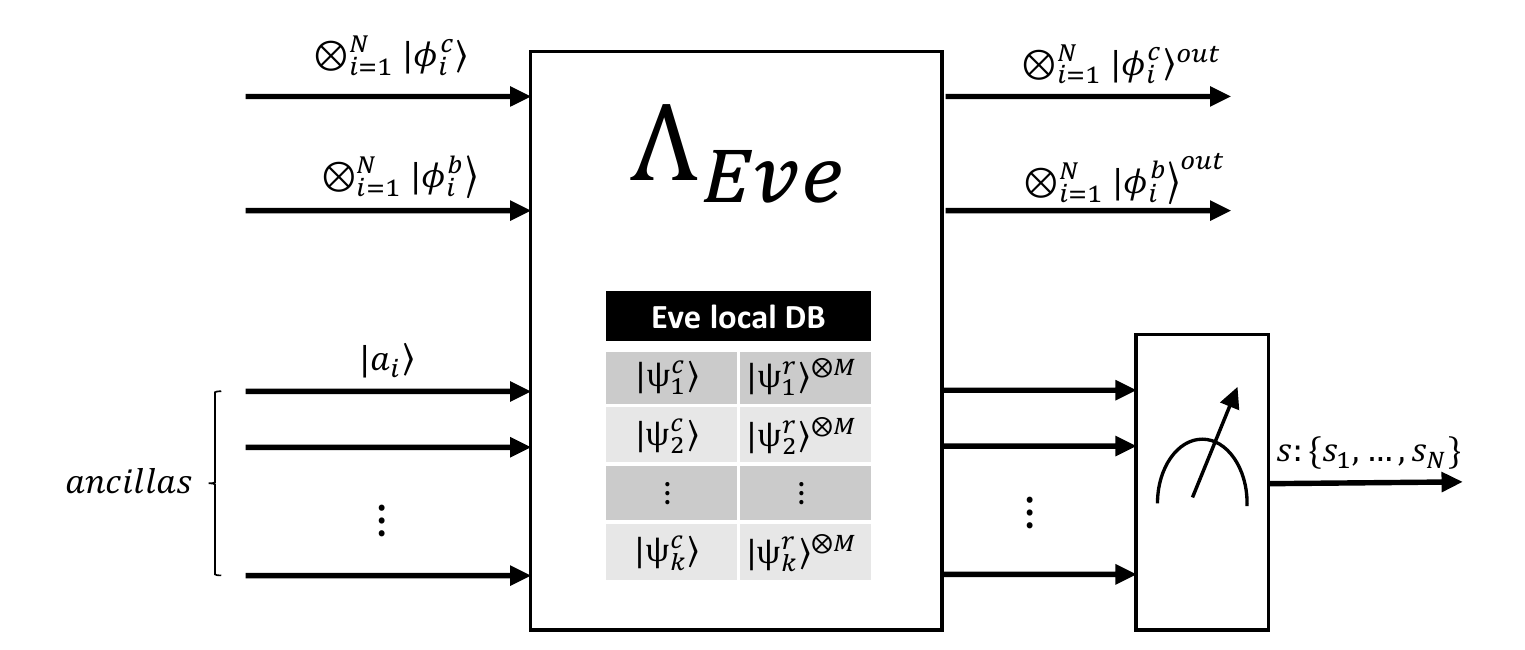}
\centering
\caption{Quantum coherent attack strategy performed by Eve on {\qp} protocol by applying the general local-database-dependent, CPTP map on the combined $N$ challenges and response states $\kc$ and $\kb$ respectively. The output is the $N$ bit string $s: \{ s_1,\cdots,s_N\}$.}
\label{fig:Figure5}
\end{figure}

We say that the {\qp} protocol is secure against any QPT Eve who performs the map $\Lambda^N$ if the resulting success probability of correctly guessing the $b$ value for all the $N$ positions is negligibly small in the security parameter.

\begin{theorem}[\textbf{Security against coherent attack}] \label{th:cohattack} The success probability of any QPT adversary in correctly guessing the $b$ values for all the $N$ positions, denoted by $\{b_1,\cdots,b_N\}$ is,
{\normalfont 
\begin{equation}
    \text{Pr}[\textbf{b} \leftarrow \Lambda^N(\ket{\phi^{\textbf{c}}},\ket{\phi^{\textbf{b}}})] \leqslant \Big(\frac{1}{2} + \mathcal{O}(2^{-d})\Big)^N 
\end{equation}
}

where $\textbf{b}: \{b_1,\cdots,b_N\}$ are the bits corresponding to correct $b$ values, $\ket{\phi^{\textbf{c}}} = \otimes_{i=1}^{N}\kc$, $\ket{\phi^{\textbf{b}}} = \otimes_{i=1}^{N}\kb$ and \normalfont{$d = \mathcal{O}(\text{poly} \log D)$} is the size of Eve's database.
\end{theorem}

\begin{proof}
To prove this theorem, we notice that Eve applies a generalised map $\Lambda^N$ on the challenge and the response states of Alice to be able to correctly distinguish whether the response states are $\kb = \kr$ for all $i \in [N]$. Thus the probability to correctly guess $\textbf{b}$ reduces to Eve applying a CPTP map $\Lambda^N$ to perform an optimal test to distinguish the response state $\ket{\phi^{\textbf{b}}}$ with her reference  state $\rho_e^N$, where $\rho_e^N$ is the generalised entangled state. Thus without loss of generality, any attack map $\Lambda^N$, consists of two parts. The first part uses a generator algorithm $gen_N$ to generate a reference state $\rho_e^N$ by using the local database and the input challenge state $\ket{\phi^{\textbf{c}}}$, and the second part performs a test algorithm $\T$ on $\ket{\phi^{\textbf{b}}}$ and $\rho_e^N$,
\begin{equation}
    \Lambda^N = \T(\ket{\phi^{\textbf{b}}}, \rho_e^N \leftarrow gen(DB,\ket{\phi^{\textbf{c}}}))
\end{equation}
where $DB$ is the local database of Eve generated in the  \emph{setup phase}. Similar to the collective strategy proof, we assume Eve's testing algorithm $\T$ is the optimal test equality test algorithm, also referred as ideal test algorithm in  Definition~\ref{def:ideal-test}, i.e. $\T = \Ti$. Here $\Ti$ again relates to the fidelity distance between the two states,
\begin{equation}
     \text{Pr}[1 \leftarrow \Lambda^N(\ket{\phi^{\textbf{c}}},\ket{\phi^{\textbf{b}}})] = \text{Pr}[1 \leftarrow \Ti] = F(\ket{\phi^{\textbf{b}}}, \rho_e^N)^2 
\end{equation}
Since each $b$ across the $N$ positions are chosen independently and randomly, this implies at entangling the map across different rounds does not help Eve in any way. Thus to correctly guess the $b$ values for all the $N$ positions, the optimal attack strategy of Eve is to generate the reference state $\rho_{max}^{\otimes N}$, such that,
\begin{equation}
    F(\rho^{\text{max}}, \kr) = \sqrt{\bra{\phi^b_i}\rho^{\text{max}}\kb} \geqslant \sqrt{\bra{\phi^b_i}\rho_{i}\kb \hspace{2mm}} \forall{i \in [N]} 
\end{equation}
where $\rho_{i} = \text{Tr}_{\{1\cdots N/i\}}(\rho_e^N)$, i.e. $\rho_{i}$ is obtained by tracing out the N-1 instances $\{1,\cdots N/i \}$. 

This further implies that attack map $\Lambda^N$ is reduced to $\Lambda_{ind}^{\otimes N}$, where the map $\Lambda_{ind}^{\otimes N}$ involves a generator algorithm that produces the state $\rho^{max}$ which maximises the average fidelity with Alice's response state across all the $N$ rounds. This implies that,
\begin{equation}
    \begin{split}
        \text{Pr}[\{b_1,\cdots,b_N\} \leftarrow \Lambda^N(\ket{\phi^{\textbf{c}}},\ket{\phi^{\textbf{b}}})] &=  \prod_{i=1}^{N}\text{Pr}[b_i \leftarrow \Lambda_{ind}(\kc,\kb)] \\
        &\leqslant \Big(\frac{1}{2} + \negl(\lambda)\Big)^N 
    \end{split}
\end{equation}
where we used the result of theorem~\ref{th:colattack} after the reduction from coherent to the collective attack. This completes the proof.
\end{proof}

\subsubsection{Comparing Classical and Quantum Strategies:}
Using the above theorems~\ref{th:colattack} and ~\ref{th:cohattack} we show that a QPT Eve does not have any non-negligible advantage in passing the \texttt{cVer} verification test compared to the purely classical Eve. Thus, we can bound the success probability of a general QPT Eve which the success probability of the classical Eve from the thereom~\ref{th:cv-clattack},
\begin{equation}
\text{Pr}[\text{Ver accept}_{\text{QPT Eve}}]  \leqslant \text{Pr}[\text{Ver accept}_{\text{Classical Eve}}] + \mathcal{O}(2^{-N}) \approx \mathcal{O}(2^{-N})
\end{equation}
%

\subsection{Protocol generalisation to arbitrary distribution of traps}\label{sec:qp-general}
In the original {\qp} protocol, Alice randomly picks half of the $N/2$ positions, and marks them $b = 1$. The rest is marked $b = 0$.  Here, even though an adversary Eve does not know the locations of valid qPUF response states and the trap states, she knows that half of the positions are traps. In this section, we generalise the {\qp} protocol to further hide the number of traps information from Eve. This is done with the hope that hiding the number of trap and good response states could further decrease the probability of Eve passing the \texttt{cVer} test especially against a fully classical Eve who only uses the statistics information to attack the protocol. Here Alice chooses an arbitrary number of trap positions. In other words, she randomly pics a value $p \in [0,1]$, then randomly picks $pN$ locations out of $N$ and marks them b= 1 (valid response states). The rest of $(1-p)N$ positions are assigned $b = 0$ (trap positions). 
One can observe that the protocol on Bob's side does not depend on this value $p$, hence Alice is not required to make the $p$ value public. 
We note that $b=1$ positions must all have bits valued 0, and $b = 0$ positions must have half bits valued 0 and the rest are valued 1 (assuming $\delta_{er} = 0$ for simplicity) if the $N$ bits have to pass the classical verification algorithm \texttt{cVer}. Now, upon running the {\qp} protocol, there are in total $N(1+p)/2$ number of  0 bits and $N(1-p)/2$ number of `1' bits in the desired bit-string $S_N$ which can pass the verification. 
Changing the tolerance value $\delta_{er}$ will not affect the result as we have seen in the previous section that by having a $\delta_{er}$ much smaller than $N$ the probability only multiplies to a constant factor.
We follow the same argument as in the proof of Theorem~\ref{th:cv-clattack}, for finding the optimal success probability of Eve generating successful bit-strings for the new classical verification. We say that the optimal strategies are the ones where their string space consists of exactly $c_1$ bits that are 1, where here $c_1=N(1-p)/2$. For the specific case of $p=0.5$, we have proven the optimality of such strategies. Hence in this specific case, we can refer to the same proof. In the generalised setting, the $p$ value is unknown, and as a result $c_1$ is unknown to Eve as well. Therefore the overall winning probability of Eve will depend on first guessing the correct values of $c_1$ and then the probability of such strings passing both tests. Also, we know that the probability of any strings with incorrect $c_1$ is necessarily 0, hence we can write the probability that Eve passes the verification test as follows,
\begin{equation}
\small
\text{Pr}[\text{Ver accept}_{\text{Eve}}] = \text{Pr}[\text{guess } c_1]\times \text{Pr}[\text{Ver accept}_{\text{Eve},S_{gop}} | c_1=\frac{N(1-p)}{2}] = \text{Pr}[\text{guess } c_1]\times \frac{{N-Np\choose\frac{N-Np}{2}}}{{N\choose\frac{N-Np}{2}}}    
\end{equation}

Let us assume that Alice, in order to maximize the randomness over the correct choice of $c_1$, picks $p$ completely uniformly from $[0,1]$. In this case, the number of trap responses can be any number between 0 (for $p=1$) and N (for $p=0$). Consequently, $c_1 \in \{0, 1,\dots, \frac{N}{2}\}$ and if any of these values occur with equal probability, then Eve can guess $c_1$ with the following probability:
\[\text{Pr}[\text{guess } c_1] = \frac{1}{\frac{N}{2} + 1}\]

Now one can calculate the average wining probability of Eve over p:
\begin{equation}
\underset{p}{\text{Pr}[\text{Ver accept}_{\text{Eve}}]} = \int_{0}^{1} \frac{2}{N + 2}\frac{(N-Np)!(\frac{N+Np}{2})!}{N!(\frac{N-Np}{2})!} dp    
\end{equation}

In the Appendix~\ref{ap:avgprob}, we have shown that the above integral converges to the following value:
\begin{equation}
\underset{p}{\text{Pr}[\text{Ver accept}_{\text{Eve}}]} \approx \overline{Pr_{win}} = \frac{2}{N + 2} \sum^{N}_{k=0} \frac{(N-k)!(\frac{N+k}{2})!}{N!(\frac{N-k}{2})!} \approx \frac{6}{N(N+2)} = \mathcal{O}(\frac{1}{N^2})    
\end{equation}

This means that by choosing the $p$ form a uniform distribution, the average success probability of the adversary becomes polynomially small in $N$ which reduces the security of the protocol to polynomial. This may seem a surprising result although the reason is that the probability function for $p=0$ and $p=1$ is 1.
On the other hand, from the security result for $p=\frac{1}{2}$, we know that the probability function's behaviour can be inverse exponential. This gives rise to the interesting question of whether one can find a boundary for $p$ in which $\text{Pr}[\text{Ver accept}_{\text{Eve}}]$ is negligible. Before addressing this problem, it is worth mentioning that by hiding $p$, one can hope the protocol's security to be boosted by at most a polynomial factor ($\frac{1}{\mathcal{O}(N)}$) as Eve's probability of guessing the correct $c_1$ depends only on the different number of 1's in the string that results from different choices of $p$. Even though for large $N$ this polynomial factor can be ignored, assuming that Alice has a good choice of $p$ which leads to the exponential security, in relatively smaller $N$ the hiding can practically boost the security of the identification.

Now to be able to analyse the $\text{Pr}[\text{Ver accept}_{\text{Eve}}]$, we rewrite the factorials with Gamma function and we define $z=\frac{N-Np}{2}$ where $z \in \{0,1,\dots,\frac{N}{2}\}$. Considering that $\Gamma(z+1)=z\Gamma(z)$, the probability is,
\[\text{Pr}[\text{Ver accept}_{\text{Eve}}] = \frac{(N-Np)!(\frac{N+Np}{2})!}{N!(\frac{N-Np}{2})!} = \frac{\Gamma(2z+1)\Gamma(N-z+1)}{N!\Gamma(z+1)} = \frac{2}{N!}\times\frac{\Gamma(2z)\Gamma(N-z+1)}{\Gamma(z)}\]

Using properties of Gamma functions we have that $\frac{\Gamma(2z)}{\Gamma(z)} = \frac{2^{2z-1}}{\sqrt{\pi}}\Gamma(z+\frac{1}{2})$. Thus we can simplify the function to be:
\begin{equation}
\text{Pr}[\text{Ver accept}_{\text{Eve}}] = \frac{2}{\sqrt{\pi}}\times\frac{2^{2z-1}}{N!}\Gamma(z+\frac{1}{2})\Gamma(N-z+1) \approx \frac{2^{2z-1}}{N!}\Gamma(z+\frac{1}{2})\Gamma(N-z+1)    
\end{equation}

\begin{figure}[t]
    \includegraphics[width=1\textwidth]{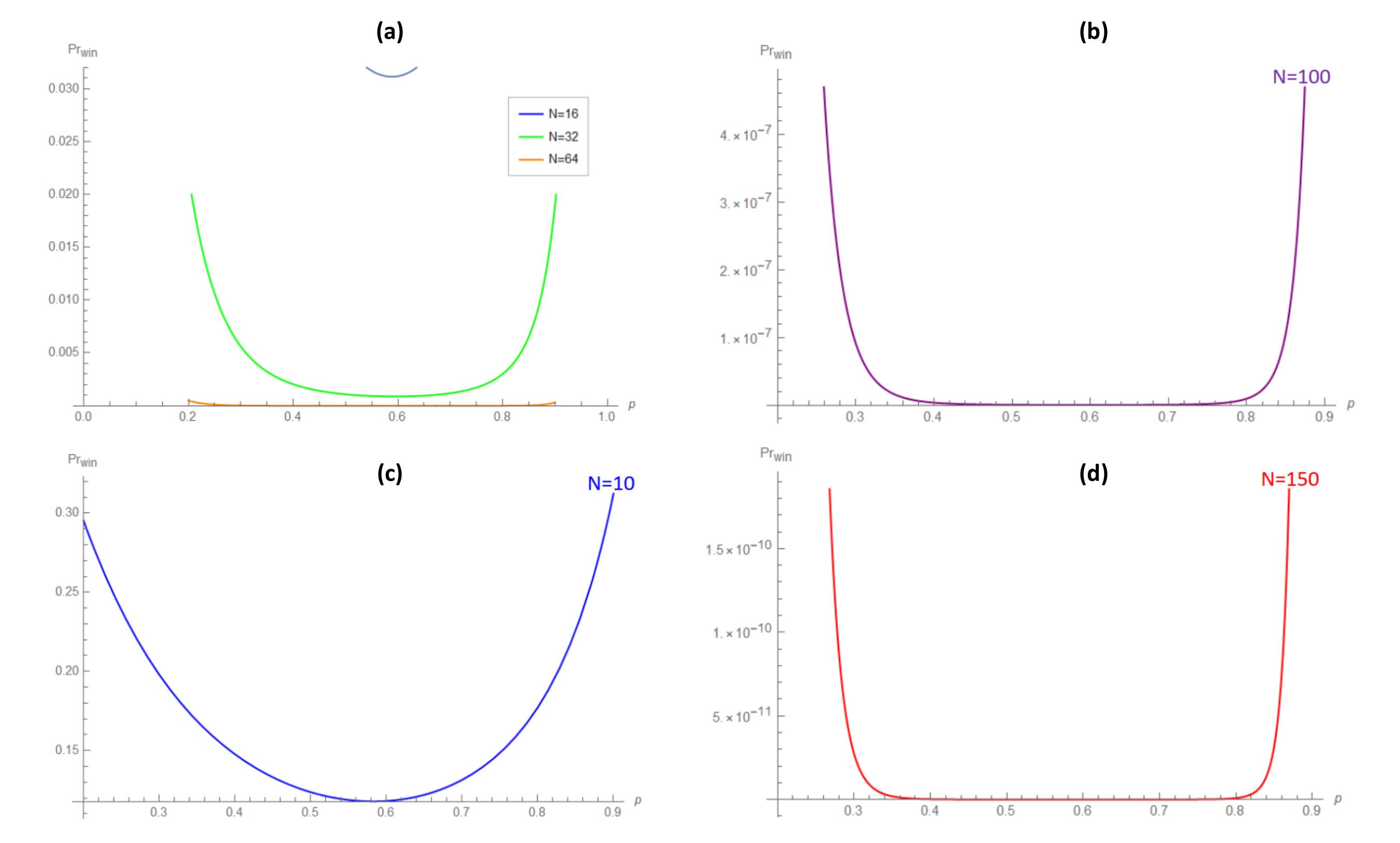}
    \caption{Behaviour of Eve's success probability $\text{Pr}[\text{Ver accept}_{\text{Eve}}]$ as a function of $p$ (corresponding to number of valid qPUF responses), for different values of $N$.}\label{fig:prob-4pic}
\end{figure}

For a large enough fixed $N$, the factor $\frac{2^{2z-1}}{N!} \ll 1$. However it is an increasing function in $z$ and $\Gamma(z+\frac{1}{2})\Gamma(N-z+1)$ is a large factor which quickly decreases with $z$. Also at the beginning and the end of the period where $z=0, z=\frac{N}{2}$, the probability is 1, and it reduces to a small value for certain $z$. Thus it can be deduced that the function will necessarily have a minimum for any $N$. The Figure~\ref{fig:prob-4pic}, different $\text{Pr}[\text{Ver accept}_{\text{Eve}}]$ for different $N$ has been shown. We have renormalised the probabilities as a function of $p$ to be able to compare them. As can be seen, the function for all the different values of $N$ falls exponentially in a minimum region where there are the desirable values of $p$. As $N$ grows, the range of desirable $p$ expands, which can be seen in the top right plot where we compare the probability for $N=16, N=32$ and $N=64$. Also by comparing the probability range for $N=10, N=100, N=150$ one can see how the exponential security is achieved for a $p$ which has been chosen in the \emph{good} region. This specification of the success probability would be useful for Alice to be able to optimise the protocol based on her resources. Moreover, the freedom of choosing traps according to desired distribution, conditioning that it bounds the value of $p$ to the minimum region, enables the protocol to be useful in other scenarios.

\section{Resource comparison of protocols} \label{sec:comparison}
\begin{table}[h!]
\centering
\resizebox{\textwidth}{!}{
\begin{tabular}{|c|c|c|c|c|c|c|c|c|}
\hline
Protocol & \multicolumn{2}{l|}{Security} & \multicolumn{2}{l|}{Quantum Memory} & \multicolumn{2}{l|}{Verification computing ability} & \multicolumn{2}{l|}{Communication round} \\ \hline
  &                    &  & Verifier & Prover & Verifier & Prover & Quantum & Classical \\ \hline
{\cps} & \multirow{3}{*}{$\epsilon$} &  = $2^{-MN}$ & $\log 1/\epsilon$      & 0     & $poly \log D$       &   0   & $\log 1/\epsilon$      & 0        \\ \cline{1-1} \cline{3-9} 
{\cpg} &                    &  = $(M+1)^{-N}$ & $\frac{M}{\log M+1}\log 1/\epsilon$       & 0     & $poly \log MD$       & 0     & $\frac{1}{\log M+1}\log 1/\epsilon$      & 0        \\ \cline{1-1} \cline{3-9} 
{\qp} &                    & = $2^{-N}$ & $\log 1/\epsilon$       & 0     & 0       & $poly \log D$     & $\log 1/\epsilon$      & 1        \\ \hline
\end{tabular}}
\caption{Comparison of different qPUF-based identification protocols in terms of security ($\text{Pr}[\text{Ver accept}_{\text{Eve}}] = \epsilon$) against any QPT adversary and the three resource categories of the verifier and the prover: quantum memory, computing ability and number of communication rounds. Here all the resources are in $\mathcal{O}(.)$. All our proposed protocols exhibits $\epsilon$  exponential security with polynomial sized resource $\mathcal{O}(\log 1/\epsilon)$ memory/communication and $\mathcal{O}(poly \log D)$ computing ability in both the parties. Here $D$ is the size of qPUF.}
\label{table:comp}
\end{table}
The two proposed qPUF-based identification protocols differ a great deal in terms of the type and amount of resources available to the concerned parties. We divide the resources into three categories: quantum memory, quantum computing ability, and the number of communication rounds required to achieve identification. Here, quantum memory is quantified by the number of quantum states stored in a register, and the computing ability resource is quantified in terms of the number of quantum gates required to implement a specific quantum circuit.

Table~\ref{table:comp} compares the resources of the two protocols that we have introduced. For a fair comparison between the above protocols, we fix the maximum acceptance probability for any QPT adversary, $\text{Pr}[\text{Ver accept}_{\text{Eve}}]$, to be $\epsilon$, and compute the number of resources required to achieve that desired acceptance probability. In all the protocols, we assume that during one identification, $N$ copies of different states, each with $M$ identical copies are used. For the specific case of {\qp} protocol, $M = 1$. For the {\cps} protocol, where the quantum verification is via the SWAP test circuit, the adversary's acceptance probability is $\epsilon = \mathcal{O}(2^{-MN})$. In this protocol, the verifier requires $MN = \mathcal{O}(\log 1/\epsilon)$ size quantum memory and computing ability of $\mathcal{O}(poly \log D)$ quantum gates, where $D$ is the size of qPUF. The prover, on the other hand, requires no quantum memory and computing ability. The number of communication rounds required to achieve the desired security is $MN = \mathcal{O}(\log 1/\epsilon)$. The protocol {\cpg}, where the quantum verification is via the GSWAP test circuit, the adversary's acceptance probability is $\epsilon = \mathcal{O}((M+1)^{-N})$. In this protocol, the verifier requires $MN = \mathcal{O}(\frac{M}{\log M+1} \log 1/\epsilon)$ size quantum memory and a computing ability of $\mathcal{O}(poly \log MD)$ quantum gates. Similar to {\cps}, the prover requires no quantum memory and computing ability. The number of communication rounds required to achieve the desired security is $N = \mathcal{O}(\frac{1}{\log M+1} \log 1/\epsilon)$. Thus for large $M$ values, the verifier's quantum memory requirement is less while using SWAP compared to GSWAP, but the number of communication rounds is higher using the SWAP test.

Now for the {\qp} protocol, the protocol with the low-resource verifier, the adversary's acceptance probability is $\epsilon = \mathcal{O}(2^{-N})$. In this protocol, the verifier requires $N = \mathcal{O}(\log 1/\epsilon)$ size quantum memory. Since the verifier performs classical verification, hence she does not require a quantum computing ability. The prover here requires no quantum memory but since he performs the SWAP test circuit, his computing ability is required to be $\mathcal{O}(poly \log D)$. The number of quantum communication rounds required to achieve the desired security is $N = \mathcal{O}(\log 1/\epsilon)$. This protocol also requires a single round of classical communication transmitting $N$ bits.

Figure~\ref{fig:all-resources} demonstrates the graphical comparison of different resources among the three qPUF-based identification protocols. The plots show a tradeoff in resources between different protocols to achieve the desired success probability of $\epsilon$. We choose the $\epsilon$ to range from $10^{-6}$ to $10^{-1}$. Since the computing ability resource depends on the qPUF size $D$, we choose $D = 1/\epsilon$ for comparison. 

We identify that the difference in resources primarily comes about due to the different requirements of SWAP and GSWAP tests. To illustrate this graphically, we provide density plots in
Figure~\ref{fig:swap-gswap-density} to showcase the trade-off between the success probability $\epsilon$ and the memory and communication round resources required for different $M$ ad $N$'s for protocols based on SWAP vs GSWAP tests. 

\begin{figure}[t]
\includegraphics[width=0.98\textwidth]{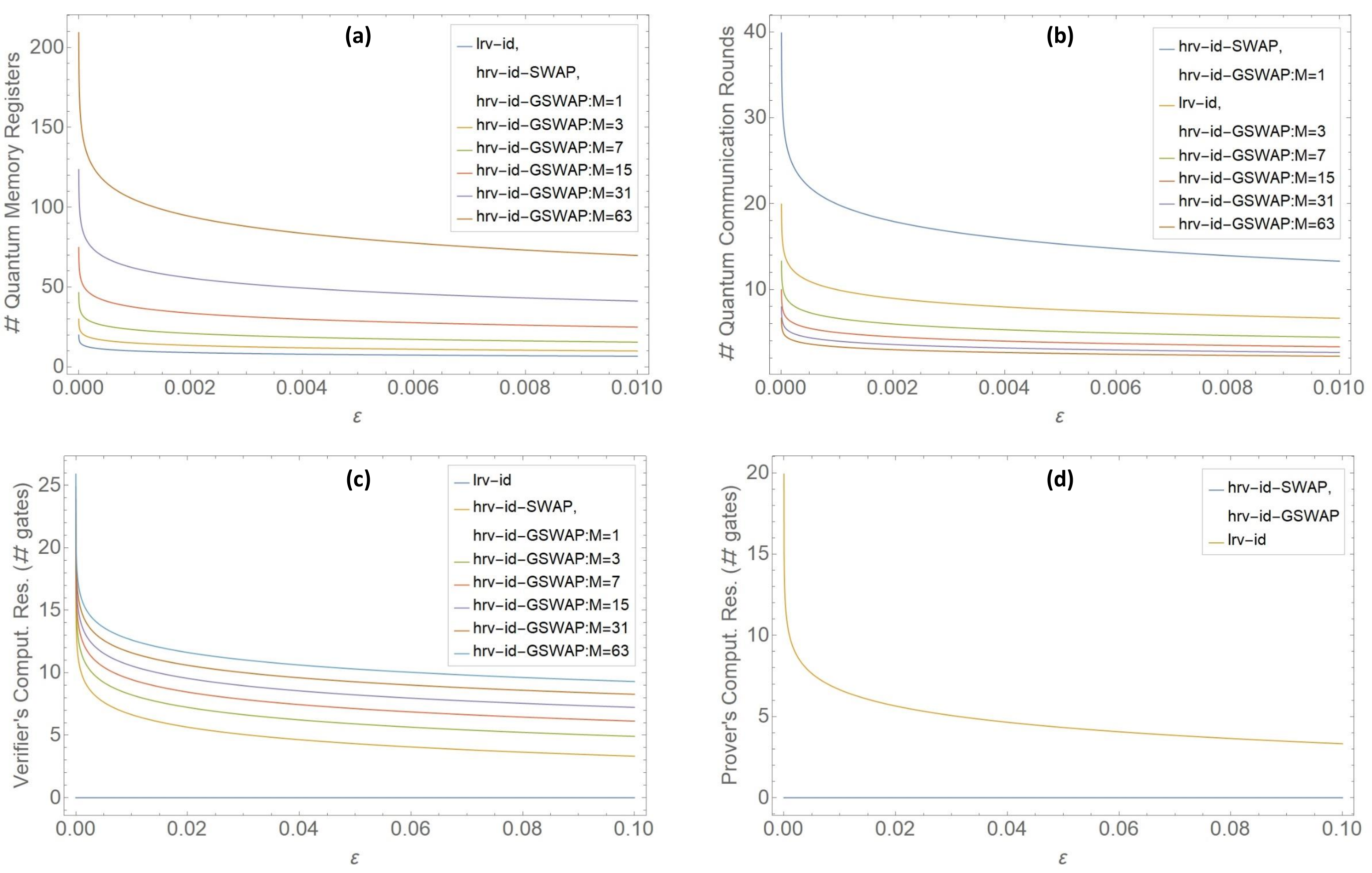}
\centering
\caption{Comparison of the resources required by the prover and verifier in the three qPUF-based identification protocols (\cps, \cpg, and \qp) for varying security values $\epsilon$. We choose the $\epsilon$ to range from 
$10^{-6}$ and $10^{-2}$ for the top row and between $10^{-6}$ and $10^{-1}$ for the bottom row. 
Plot~top left compares the verifier's quantum memory resource vs $\epsilon$ for the three protocols. The plot shows that the least memory requirement is minimum in {\cps} and {\qp} protocols while it increases by increasing the number of local copies $M$ required in the GSWAP test for {\cpg} protocol. We note that the prover's memory requirement is 0 in all the three protocols. Plot~top right similarly compares the number of quantum communication rounds in the three protocols. The requirement is minimum in the {\qp} while it increases with $M$ in the {\cpg}. The communication round in {\cpp} is double to the {\qp} requirement to indicate the two-way quantum communication instead of one way in the latter. 
Plots~bottom left and bottom right compares the computational resource vs $\epsilon$ for the verifier and prover respectively. Here we have taken $D = 1/\epsilon$ for comparison. 
}  
\label{fig:all-resources}
\end{figure}

\begin{figure}[t]
\includegraphics[width=0.98\textwidth]{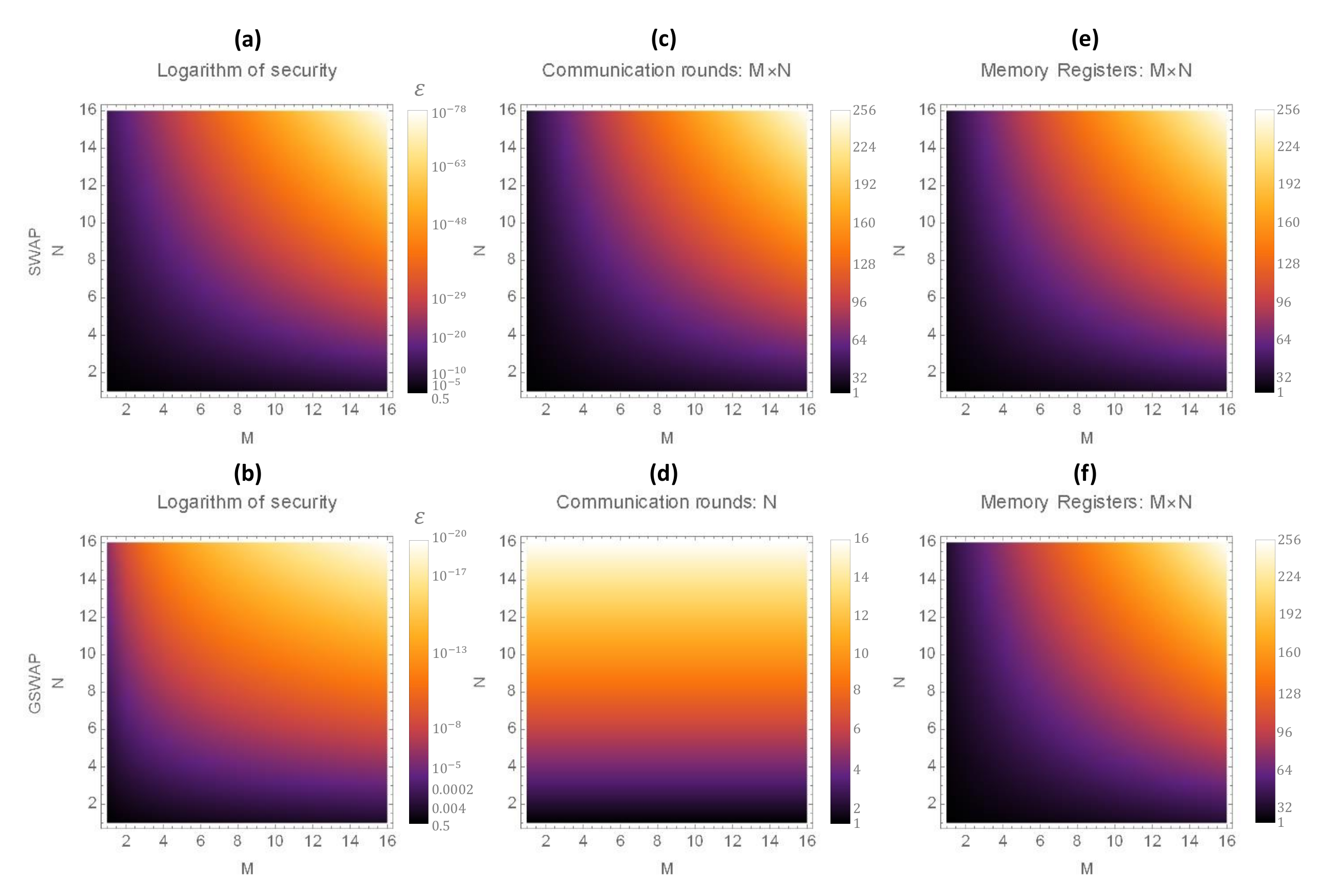}
\centering
\caption{Comparison of verification based on SWAP and GSWAP for identification protocols. The top row is associated with SWAP and the button row to GSWAP. The x-axis of the plots are all $M$ (the number of local copies) and the y-axis are all $N$ (the number of different states) and the security, quantum memory and quantum communications have been shown with colour spectrum. The left column shows the security $\epsilon$ where we have $\epsilon = 2^{-MN}$ for SWAP and $\epsilon = (M+1)^{-N}$ for GSWAP, in a logarithmic scale for more visibility. The middle column shows the required communication where we see that for GSWAP the communication rounds are independent of $M$ and only linearly growing with $N$ while as for SWAP the communication rounds grows also linear by increasing the number of local copies. The right column shows the memory which has been fixed for both SWAP and GSWAP to $M\times N$. The comparison between security and communication plots shows a trade-off between SWAP and GSWAP as the quantum verification algorithm.}  
\label{fig:swap-gswap-density}
\end{figure}

\section{Discussion} \label{sec:discussion}

We have proposed two qPUF based identification protocols which provide exponential security against any QPT adversary by solely utilising hardware-based qPUF property instead of other cryptographic properties of the device. Our primary classification in the two protocols have come about from the practical scenarios in a network, i.e. parties with varying capabilities should be able to efficiently run a secure identification protocol.
Our first protocol, {\cpp}, is proposed to be suited more in the mobile-like device settings i.e. provers having low resource would want their device to be correctly identified by a high resource verifier. Since the identification protocol requires a multi-round communication between the prover and the verifier, we have proposed efficient quantum equality-testing based verification tests to reduce the communication overhead requirement.

Our second protocol, {\qp}, is suited in the mobile-like verification setting i.e. a low-resource almost classical-like verifier would want to verify the device of a high resource quantum prover. The advantage of this protocol is that a purely classical verification algorithm is sufficient to verify the prover's device with provable security. 
{\qp} is based on the idea of verifier inserting random trap states in between the communication rounds which facilitates a secure delegation of the quantum testing to the prover. This allows the verifier to simply run a classical algorithm on the quantum test outcomes to perform successful identification. 

An interesting extension of {\qp} protocol that we have shown is the generalisation to the arbitrary distribution of traps instead of randomly inserting them in half the positions as proposed in our current protocol. With this generalisation on hiding the trap distribution, one hopes for further enhancement in security against a QPT adversary. We draw some non-trivial conclusions from this generalisation including the worsening of security to polynomial in the number of communication rounds (instead of exponential as our current protocol) when the number of trap positions is chosen uniformly over the total positions. We also remark that some distributions provide a polynomial enhancement over the current exponential security bound, thus justifying the need for hiding the number of trap positions.

\section*{Acknowledgement}

We acknowledge the support of the European Unions Horizon 2020 Research
and Innovation Programme under Grant Agreement No. 820445 (QIA) and the UK Engineering and Physical Sciences Research Council grant  EP/N003829/1.

\section*{Author Contributions}

M.Do. and N.K. did the proofs. M.Do. did the plots and simulations. All authors contributed to preparing the manuscript.  

\section*{Competing Interests}

The authors declare no competing interest.


\bibliographystyle{ACM-Reference-Format}
\bibliography{clientserverqpuf-acm}

\appendix
\section{Appendix}

\subsection{ Selective unforgeability of unknown unitary qPUFs}\label{sec:sel-unf}
We take the results in ~\cite{arapinis2021quantum} for the selective unforgeability of the qPUF to be able to use it in our security proofs. First, we restate a theorem which implies that the success probability of any QPT adversary to output the response of a Haar random challenge state $\rho \in \HilD$ with non-negligible fidelity is bounded.

\begin{theorem}\label{th:sel-qPUF}\textbf{[restated from~\cite{arapinis2021quantum}]} For any unitary evolution $U$, picked from an unknown unitary family, for any non-zero $\delta$ and any state $\rho \in \HilD$ randomly picked from Haar measure over $\HilD$, the success probability of any QPT algorithm $\A$ having a $polylog(D)$-size pre-challenge access to $U$, in outputting a state closer than $\delta$ in the fidelity distance to state $U\rho U^{\dagger}$ is bounded as follows:
\begin{equation}
    \underset{\rho\in\HilD}{\text{Pr}}[F(\A(\rho),U\rho U^{\dagger})^2 \geqslant \delta] \leqslant \frac{d + 1}{D}    
\end{equation}
where $D$ is the dimension of the Hilbert space that the challenge quantum state is picked from, and $0\leq d \leq D-1$ is the dimension of the largest subspace of $\HilD$ that the adversary can span during the learning phase.
\end{theorem}
\begin{proof}
Here we state the proof of above theorem according to the original proof given in~\cite{arapinis2021quantum}. Let $\A$ create an input and output database from querying $U$ namely $\Sin$ and $\Sout$, both with size $k$. Also, Let $\Hild$ be the $d$-dimensional Hilbert space spanned by elements of $\Sin$ and $\Hildout$ be the Hilbert space spanned by elements of $\Sout$ with the same dimension. As $U$ is an unknown unitary $\A$ only learns $U$ through queries thus $d \leq k$. $\A$ receives an unknown quantum state $\rho$ as a challenge and tries to output a state $\omega$ or its purification $\ket{\omega}$ as close as possible to $\rho^o = U\rho U^{\dagger}$. The objective is to bound the average probability of $\A$'s output state $\ket{\omega}$ to have a fidelity larger or equal to $\delta$ such that for any $\delta \neq 0$ the success probability will be negligible. The average probability is over all the possible states of $\rho \in \HilD$ picked at random from a uniform distribution (Haar measure). Thus we are interested in the following probability:
\begin{equation}
\begin{split}
        Pr_{success} =& \underset{\rho\in\HilD}{Pr}[F(\A(\Sin, \Sout, \rho), U\rho U^{\dagger})^2 \geq \delta] =\\
        & \underset{\rho\in\HilD}{Pr}[|\bra{\omega}\rho^o\ket{\omega}|^2 \geq \delta].
\end{split}
\end{equation}
The $\Hild$ is a known subspace for $\A$ as by picking the input states $\A$ can have the classical descriptions. Although the $\Hildout$ is an unknown subspace. Here we boost $\A$ to be a stronger adversary by assuming that $\A$ gets access to the complete set of basis of $\Hild$ and $\Hildout$ or in other words the complete description of the map in the subspace.
Let $\{\ket{e^{in}_i}\}^d_{i=1}$ and $\{\ket{e^{out}_i}\}^d_{i=1}$ be the sets of orthonormal basis of the input and output subspaces. 
Now, we partition the set of all the challenges to two parts: the challenges that are completely orthogonal to $\Hild$ subspace, and the rest of the challenges that have non-zero overlap with $\Hild$. We denote the subspace of all the states orthogonal to $\Hild$ as $\Hildperp$. In other words, we will analyse the target probability $Pr_{success} = \underset{\rho\in\HilD}{Pr}[|\bra{\omega}\rho^o\ket{\omega}|^2 \geq \delta]$ in terms of the partial probabilities
\begin{equation}
\begin{split}
        &\underset{\rho\in\HilD, \rho\in\Hildperp}{Pr}[|\bra{\omega}\rho^o\ket{\omega}|^2 \geq \delta] \quad \text{and}\\
        &\underset{\rho\in\HilD, \rho\not\in\Hildperp}{Pr}[|\bra{\omega}\rho^o\ket{\omega}|^2 \geq \delta].
\end{split}
\end{equation}
Because the probability of $\rho$ being in any particular subset is independent of the adversary's picked subspace, the success probability can be written as:
\begin{equation}
\begin{split}
        Pr_{success} = & \underset{\rho\in\Hildperp}{Pr}[|\bra{\omega}\rho^o\ket{\omega}|^2 \geq \delta]\times Pr[\rho \in \Hildperp] + \\
        & \underset{\rho\not\in\Hildperp}{Pr}[|\bra{\omega}\rho^o\ket{\omega}|^2 \geq \delta]\times Pr[\rho \not\in \Hildperp]
\end{split}
\end{equation}
where $Pr[\rho \in \Hildperp] = 1-Pr[\rho \not\in \Hildperp]$ denotes the probability of $\rho$ picked accorging to Haar measure, being projected into the subspace of $\Hildperp$. Now we refer to Lemma 1 in~\cite{arapinis2021quantum} stating that this probability for any subspace, is equal to the ratio of the dimensionalities. As $\Hildperp$ is a $D-d$ dimensional subspace, $Pr[\rho \in \Hildperp] = \frac{D-d}{D}$ and respectively $Pr[\rho \not\in \Hildperp] = \frac{d}{D}$. Also the probability is upper-bounded by the cases that the adversary can always get a good fidelity for $\rho \not\in \Hildperp$:
\begin{equation}
    Pr_{success} \leq \underset{\rho\in\Hildperp}{Pr}[|\bra{\omega}\rho^o\ket{\omega}|^2 \geq \delta]\times(\frac{D-d}{D}) + \frac{d}{D}
\end{equation}
Finally it only remains to bound the success probability of $\A$ over the subspace completely orthogonal to the learnt one. Any state $\ket{\omega}$ produced by $\A$ can be written in the following form
\begin{equation}
    \ket{\omega} = \sum^d_{i=1}\beta_i\ket{e^{out}_i} + \sum^D_{i=d+1}\gamma_i\ket{q_i} 
\end{equation}
where the first part is spanned by the basis of learnt output subspace and the second part has been produced in $\Hildperpo$ with $\{\ket{q_i}\}^{D-d}_{i=1}$ being a set of bases for $\Hildperpo$. For all $\rho\in\Hildperp$ as the unitary preserve the inner product the output $\rho^o$ is also orthogonal to $\Hildout$. Thus the first part of the state $\ket{\omega}$ always gives a $0$ fidelity and for $\A$ to optimise the probability all $\beta_i$ should be zero. This leads to all adversaries states be in the form of $\sum^{D-d}_{i=1}\gamma_i\ket{q_i} \in \Hildperpo$ where the normalisation condition is $\sum^{D-d}_{i=1}|\gamma_i|^2 = 1$. Now according to the argument given in~\cite{arapinis2021quantum}, the selection of this $\ket{q_i}$ basis is completely independent to the actual basis of $\rho$ as it has been randomly picked from Haar measure over $\HilD$. More precisely, one needs to bound the probability of the average fidelity being greater than $\delta$ for this subspace. By using the symmetry of the fidelity and Haar distributed states, it can be shown that the average can be taken over both $\rho^o$ and $\ket{\omega}$:
\begin{equation}
\small
\begin{split}
    \\ & \underset{\rho^o\Hildperpo}{\int}|\bra{\omega}\rho_x^o\ket{\omega}|^2d\mu_x = \underset{\rho^o\Hildperpo}{\int}|\sum^{D-d}_{i=1}\overline{\gamma_i}|\bra{q_i}\rho_x^o\ket{q_i}|^2d\mu_x = \\
    & \underset{\rho^o\in\Hildperpo}{\int}|\sum^{D-d}_{i=1}\overline{\gamma_{i_{x}}}|\bra{q_i}\rho^o\ket{q_i}|^2d\mu_x = \underset{\ket{\omega}\in\Hildperpo}{\int}|\bra{\omega_x}\rho^o\ket{\omega_x}|^2d\mu_x
\end{split}
\end{equation}
where $d\mu$ denotes the Haar measure. According to our uniformity assumption, the $d\mu$ here is the Haar measure. Note that $\ket{\omega}$ can be different for any new challenge. Now instead of bounding this average with $\delta$, a more general case can be considered in which this average is any non-zero quantity. As it has been shown~\cite{arapinis2021quantum}, the probability of being zero i.e. $\underset{\ket{\omega}\in\Hildperpo}{Pr}[|\bra{\omega}\rho^o\ket{\omega}|^2 = 0]$ is greater than the probability of being projected into a $D-d-1$ dimensional subspace hence we have:
\begin{equation}
\begin{split}
    & \underset{\ket{\omega}\in\Hildperpo}{Pr}[|\bra{\omega}\rho^o\ket{\omega}|^2 = 0] \geq \\
    & \underset{x}{Pr}[(\sum^{D-d}_{i,j=1}|\overline{\gamma_{i_{x}}}\alpha_j|^2|\bra{q_{i}}\Pi_j\ket{q_{i}}|) = 0] = \frac{D-d-1}{D-d}.
\end{split}
\end{equation}
Here $\alpha_i$ are coefficients for the expansion of $\rho^o$. Consequently,
\begin{equation}
\underset{\rho^o\in\Hildperpo}{Pr}[|\bra{\omega}\rho^o\ket{\omega}|^2 \neq 0] \leq \frac{1}{D-d}
\end{equation}
which also holds for any non-zero delta. Substituting this into the success probability the result will be
\begin{equation}
    Pr_{success} \leq \frac{1}{D-d}\times(\frac{D-d}{D}) + \frac{d}{D} = \frac{d+1}{D}
\end{equation}
and the theorem has been proved.
\end{proof}
\subsection{Quantum Equality Tests}\label{sec:test}
Distinguishing two unknown quantum states is a central ingredient in quantum information processing. This task is often referred to as the `state discrimination task'. The celebrated Holevo-Helstrom bound \cite{holevo1973bounds} relates the optimal state distinguishability of two unknown states with the trace distance between the states. This implies that unless the states are the same (up to a global factor), it is impossible to deterministically distinguish the two states. An important application of state discrimination is the task of Equality testing \cite{buhrman2001quantum, barenco1997stabilization, xu2015experimental}. This is an extremely simple task but a building block for lots of complicated quantum protocols. The objective of Equality testing, one that we consider in our work, is to test whether two \emph{unknown} quantum states are the same. This is a well-studied topic and we describe the optimal quantum protocols for Equality testing.

\subsubsection{SWAP test}\label{sec:swap}

Given a single copy of two unknown quantum states $\rho$ and $\sigma$, is there a simple test to optimally determine whether the two states are equal or not? This question was answered in affirmative by Buhrman et al \cite{buhrman2001quantum} when they provided a test called the SWAP test. This test was initially used by the authors to prove an exponential separation between classical and quantum resources in the simultaneous message passing model. Since then it has been used as a standard tool in the design of various quantum algorithms \cite{buhrman2010nonlocality,kumar2017efficient}. A SWAP test circuit takes as an input the two unknown quantum states $\rho$ and $\sigma$ and attaches an ancilla $\ket{0}$. A Hadamard gate is applied to the ancilla followed by the control-SWAP gate and again a Hadamard on the ancilla qubit. Finally, the ancilla is measured in the computational basis and we conclude that the two states are equal if the measurement outcome is `0' (labelled accept). Figure~\ref{fig:swap} illustrates this test in the special case when the state $\sigma$ is a pure state and shown by $\ket\psi$.
\begin{figure}[ht!]
    \centering
    \[    \Qcircuit @C=2em @R=1.4em {
       & \lstick{\ket{0}} & \gate{H} & \ctrl{1} & \gate{H} & \meter \\
       & \lstick{\rho} & \qw & \multigate{1}{\text{SWAP}} & \qw & \qw \\
       & \lstick{\ket{\psi}} & \qw & \ghost{\text{SWAP}} & \qw & \qw
    }\]
    \caption{The SWAP test circuit}
    \label{fig:swap}
\end{figure}
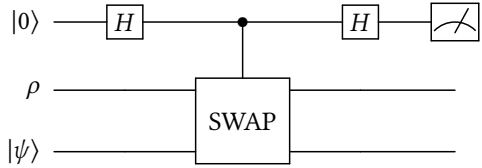

It can be shown that the probability the SWAP test accepts the states $\rho$ and $\sigma$ is \cite{kobayashi2003quantum},
\begin{equation}
    \text{Pr}[\text{SWAP accept}] = \frac{1}{2} + \frac{1}{2}\text{Tr}(\rho\sigma)
\end{equation}

In the special case of when at least one of the states (let's say $\sigma$) is a pure state $\sigma = \ket{\psi}\bra{\psi}$, the probability of acceptance is,
\begin{equation}
     \text{Pr}[\text{SWAP accept}] = \frac{1}{2} + \frac{1}{2} |\bra{\psi}\rho\ket{\psi} = \frac{1}{2} + \frac{1}{2}F^2(\rho, \ket{\psi}\bra{\psi})
     \label{eq:swapaccept}
\end{equation}

Thus when at-least one of the two states is a pure state, the acceptance probability is related to the fidelity between the states. This implies when the states are the same, the probability of acceptance is 1. However, when the states are different then if the SWAP test accepts the states, this implies an error. Thus the error in the SWAP test when the states are different (also called the one-sided error) is $\text{Pr}[\text{accept}]$. This error can, however, be brought down to any desired error $\epsilon > 0$ by running multiple instances of the SWAP test circuit. The number of instances required to bring down the error probability to a desired $\epsilon$ is,

\begin{equation}
\begin{split}
    \text{Pr}[\text{SWAP error}] & = \prod^{M}_{j=1}\text{Pr}[\text{SWAP accept}]_j = (\frac{1}{2} + \frac{1}{2}F^2)^M = \epsilon \\
    & \Rightarrow M(\log(1+F^2)-1) = \log(\epsilon) \Rightarrow M\approx \mathcal{O}(\log(1/\epsilon))
\end{split}
\end{equation}
where $F = F(\rho, \ket{\psi}\bra{\psi}) = \sqrt{\bra{\psi}\rho\ket{\psi}}$ and we use the fact that fidelity is independent of $\epsilon$. 

\subsubsection{Generalised SWAP test}\label{sec:gswap}
The above SWAP test is optimal in Equality testing (in a single instance) of two unknown quantum states when one has a single copy of the two states. However, there are certain quantum protocols where one has access to multiple copies of one unknown state $\ket{\psi}$ and only a single copy of the other unknown state $\rho$ and the objective is to provide an optimal Equality testing circuit. Considering this scenario, Chabaud et al.~\cite{chabaud2018optimal} provided an efficient construction of such a circuit, generalised SWAP (GSWAP) test circuit. A GSWAP circuit takes as an input a single copy of $\rho$, M copies of $\ket{\psi}$ and $\ceil[\big]{\log M+1}$ copies of the ancilla qubit $\ket{0}$. The generalised circuit is then run on the inputs, and the ancilla qubits are measured in the computational basis. Figure~\ref{fig:gswap} is a generic illustration of such a circuit. For more details on the circuit refer to the original work \cite{chabaud2018optimal}.
\begin{figure}[ht!]
\includegraphics[scale=0.25]{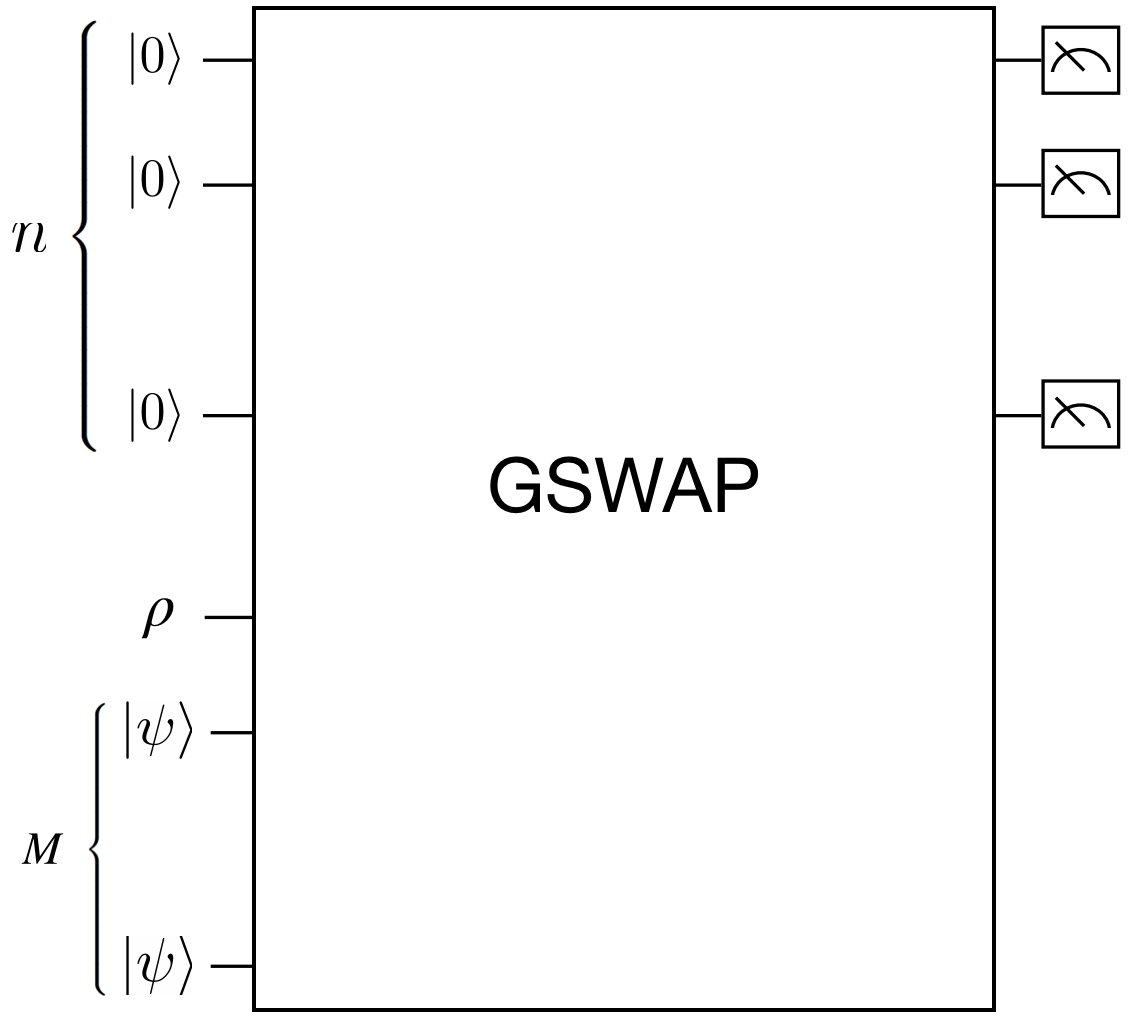}
    \centering
    \caption{GSWAP: A generalisation of the SWAP test with a single copy of $\rho$ and $M$ copies of $\ket{\psi}$. The circuit also inputs $n = \ceil[\big]{\log M+1}$ ancilla qubits in the state $\ket{0}$. At the end of the circuit, the ancilla states are measured in the computational basis.}
    \label{fig:gswap}
\end{figure}
It can be shown that the probability the GWAP circuit accepts two quantum states $\rho$ and $\ket{\psi}$ is,
\begin{equation}
     \text{Pr}[\text{GSWAP accept}] = \frac{1}{M+1} + \frac{M}{M+1} \bra{\psi}\rho\ket{\psi} = \frac{1}{M+1} + \frac{M}{M+1}F^2
     \label{eq:gswap}
\end{equation}
where $F = F(\rho, \ket{\psi}\bra{\psi})$. We note that in the special case of $M=1$, the GSWAP test reduces to the SWAP test. Also in a single instance, GSWAP provides a better Equality test compared to the SWAP test since it reduces the one-sided error probability. In the limit $M \rightarrow \infty$, we obtain the optimal acceptance probability of $\text{Pr}[\text{accept}] = F^2 = \bra{\psi}\rho\ket{\psi}$. Another important feature of GSWAP is that it can achieve any desired success probability $\epsilon (\geqslant F^2)$ in just a single instance which is impossible to achieve using SWAP circuit. However, the number of copies required is exponentially more than the number of instances that the SWAP circuit has to run to achieve the same error probability,
\begin{equation}
\begin{split}
    \text{Pr}[\text{GSWAP error}] & = \text{Pr}[\text{GSWAP accept}] = \frac{1}{M+1} + \frac{M}{M+1}F^2 = \epsilon \\
    & \Rightarrow M\approx \mathcal{O}(1/\epsilon)
\end{split}
\label{eq:gswaperror}
\end{equation}

Hence one decides the use of either SWAP test or GSWAP test depending on the specific application.

\subsubsection{Abstract and ideal quantum Equality test}\label{sec:itest}

From the tests described above, we define an abstract and ideal version of the quantum Equality test when at-least one of the states is a pure state, and relate it to the fidelity distance as discussed in \cite{arapinis2021quantum} paper. 

\begin{definition}[Quantum Testing Algorithm]\label{def:test} Let $\rho^{\otimes \kappa_1}$ and $\ket{\psi}^{\otimes \kappa_2}$ be $\kappa_1$ and $\kappa_2$ copies of two quantum states $\rho$ and $\ket{\psi}\bra{\psi}$, respectively. A Quantum Testing algorithm $\T$ is a quantum algorithm that takes as input the tuple ($\rho^{\otimes \kappa_1}$,$\ket{\psi}^{\otimes \kappa_2}$) and generates an outcome `1'(accept) when $\rho$ and $\ket{\psi}\bra{\psi}$ are equal with the probability,
\begin{equation}
\mathrm{Pr}[1 \leftarrow \T(\rho^{\otimes \kappa_1}, \ket{\psi}^{\otimes \kappa_2})] = f(\kappa_1,\kappa_2, F(\rho, \ket{\psi}\bra{\psi}))
\end{equation}
where $F(\rho, \ket{\psi}\bra{\psi})$ is the fidelity between the two states and $f(\kappa_1,\kappa_2, F(\rho, \ket{\psi}\bra{\psi}))$ satisfies the following limits:
\begin{equation}
 \begin{cases}
    \lim_{F(\rho, \ket{\psi}\bra{\psi}) \rightarrow 1}f(\kappa_1,\kappa_2, F(\rho, \ket{\psi}\bra{\psi})) = 1  & \forall\:(\kappa_1,\kappa_2)\\
    \lim_{\kappa_1 =1,\kappa_2 \rightarrow \infty}f(\kappa_1,\kappa_2, F(\rho, \ket{\psi}\bra{\psi})) = F^2(\rho, \ket{\psi}\bra{\psi})\\
    \lim_{\kappa_1 \rightarrow \infty,\kappa_2 =1}f(\kappa_1,\kappa_2, F(\rho, \ket{\psi}\bra{\psi})) = F^2(\rho, \ket{\psi}\bra{\psi})\\
    \lim_{F(\rho, \ket{\psi}\bra{\psi}) \rightarrow 0}f(\kappa_1,\kappa_2, F(\rho, \ket{\psi}\bra{\psi})) = \epsilon(\kappa_1, \kappa_2)
  \end{cases} 
\end{equation}
with $\epsilon(\kappa_1,\kappa_2)$ is the statistical error due to the Equality test algorithm.
\end{definition}
As an example, for the GSWAP test where $\kappa_1 = 1$ and $\kappa_2 = M$, we obtain from Eq~\ref{eq:gswaperror} that the probability of acceptance in the limit ${F(\rho, \ket{\psi}\bra{\psi}) \rightarrow 1}$ is 1, while it is $\frac{1}{M+1}$ in the limit ${F(\rho, \ket{\psi}\bra{\psi}) \rightarrow 0}$. It can be inferred from the above definition that the quantum test can be idealized by forcing the $\epsilon(\kappa_1,\kappa_2)$ to be zero for any given number of copies. This implies that one can abstractly construct an ideal test in a single instance case (when one is provided with a single copy of one quantum state and multiple copies of the other state),
 
\begin{definition}[Single Instance Ideal Test Algorithm]\label{def:ideal-test} We call a test algorithm according to Definition~\ref{def:test}, a $\Ti$ test algorithm when one is provided a single copy of the state $\rho$ and multiple copies of the state $\ket{\psi}$ (or vice-versa) with fidelity $F(\rho, \ket{\psi}\bra{\psi})$ the test responds as follows:
\end{definition}
\begin{equation}
 \Ti := \mathrm{Pr}[1 \leftarrow \Ti(\rho, \ket{\psi}\bra{\psi})] = F^2(\rho, \ket{\psi}\bra{\psi})
\end{equation}

\subsection{Proof of Theorem~\ref{th:cv-clattack}}\label{ap:lrv-sound}
Before proving the above theorem, we remark that any classical Eve's strategy to produce a valid $N$-bit string $S_N$ can be divided into two categories,

\begin{enumerate}
    \item \textbf{Independent guessing strategy:} Under this strategy, Eve tries to independently guess each bit of the string $S_N$ that would pass Alice's \texttt{cVer} algorithm. This also relates to the strategy of independently finding valid response and trap positions.

    \item \textbf{Global strategy:} Here,  Eve strategy is to output a string $S_N$ using the global properties of the \texttt{cVer} such that the string passes the verification test with maximum probability. In contrast to the previous strategy, the probability to output each bit $s_i$ is not necessarily independent with the global strategy.  
\end{enumerate}

We calculate the optimal success probability of Eve in both cases and show that by optimizing over both the strategies, we obtain a higher success probability for Eve in the optimal global strategy scenario. Although, the two strategies converge in the limit of large $N$. Hence we bound Eve's success probability by the optimal global strategy.\\

\noindent \textbf{1. Independent guessing strategy:} Under this strategy, Eve independently guesses each bit with the probability,

\begin{equation}
    Pr[s_i = 0] = \alpha, \quad Pr[s_i = 1] = 1 - \alpha
\end{equation}
where $\alpha \in [0,1]$.

We denote the resulting string generated by Eve's strategy as $S_{id} = \{s_1,\cdots, s_N\}$. In order for $S_{id}$ to pass the \texttt{cVer} verification algorithm, it must simultaneously pass the \texttt{test1} and \texttt{test2}. Since Eve's strategy is guessing each bit independently, hence the probability for her to pass the \texttt{test1} and \texttt{test2} are independent. Let us look at the probability of passing the \texttt{test1} (which corresponds to checking the $N/2$ positions marked $b = 1$,

\begin{equation}
        \text{Pr}[\texttt{test1 pass}] = \text{Pr}[s_{p_1} = 0]\times\cdots\times \text{Pr}[s_{p_{\frac{N}{2}}} = 0] = \alpha^{\frac{N}{2}}
\end{equation}
where $p_i$ correspond to the $b=1$ marked positions. 

If Eve's generated string passes \texttt{test1}, then Alice runs the \texttt{test2} to check if \emph{count}, which is the number of bits that are 1 in the remaining $N/2$ bits marked with $b=0$, lies within the interval $\big \lvert count - \frac{N}{4} \big\rvert \leqslant \delta_{er}$. Eve succeeds in passing this test with the probability, 
\begin{equation}
    \text{Pr}[\texttt{test2 pass}] = \sum_{x= N/4 - \delta_{er}}^{N/4 + \delta_{er}}(1-\alpha)^{x}\alpha^{\frac{N}{2} - x}\times {N/2 \choose x} \approx (2\delta_{er}+1)(1-\alpha)^{\frac{N}{4}}\alpha^{\frac{N}{4}}\times {N/2 \choose N/4} 
\end{equation}
where the approximation holds since we assume that $\delta_{er} \ll N$. From the above results, we see that the probability that Eve's string $S_{id}$ passes the \texttt{cVer} verification algorithm is,
\begin{equation}
        \text{Pr}[\text{Ver Accept}_{\text{Eve},\alpha}] = \text{Pr}[\texttt{test1 pass}_{\alpha}]\cdot \text{Pr}[\texttt{test2 pass}_{\alpha}] \approx (2\delta_{er}+1)\alpha^{\frac{3N}{4}}(1-\alpha)^{\frac{N}{4}}\times {N/2\choose N/4}
\end{equation}
This is Eve's acceptance probability for a given $\alpha$. An optimal strategy for Eve is find the optimal value of $\alpha$ that maximises the acceptance probability. This corresponds to, 

\begin{equation}
        \frac{\partial}{\partial \alpha}\text{Pr}[\text{Ver Accept}_{\text{Eve},\alpha}] \Rightarrow \frac{\partial}{\partial \alpha}(\alpha^{\frac{3N}{4}}(1-\alpha)^{\frac{N}{4}}) = 0 \Rightarrow \alpha = \frac{3}{4}
\end{equation}

Thus the maximum acceptance probability of Eve using an independent guessing strategy is:
\begin{equation}
     \text{Pr}[\text{Ver Accept}_{\text{Eve}}] = (2\delta_{er} + 1) \frac{3^{\frac{3N}{4}}}{2^{2N}}\times {N/2\choose N/4} \approx \mathcal{O}(2^{-N})
\end{equation} \\

\noindent\textbf{2. Global strategy:} The second category of Eve's strategy is to guess the $N$ bit string which passes the \texttt{cVer} test algorithm with maximum probability. Here, Eve is not restricted to choosing each bit independently.  
To find the optimal global strategy we look at the \texttt{test1} and \texttt{test2} algorithms and extract out essential properties that can be leveraged by Eve to pass the verification test. We note that
\begin{itemize}
    \item Since the good and trap response positions corresponding to $b = 0$ and 1 are chosen uniformly randomly by Alice, hence Alice does not have any information on the index set $P$ corresponding to $b = 1$ (thus no information on $b=0$ positions too).

    \item Eve knows the statistics of 0's and 1's in the desired string to pass the \texttt{cVer}. For example, a string must have a minimum of $\approx 3N/4$ bits which are $0$, otherwise, the string necessarily fails the \texttt{test1} or \texttt{test2} or both.
\end{itemize}
Based on the above facts, any global strategy for Eve should consist of optimizing the number of 0's and 1's to pass both verification tests.

Before considering the optimal global attack strategy, we give an example of a specific (non-optimal) attack strategy to provide intuition on the kind of strategies that Eve can adopt here.\\

\noindent \textbf{Example global strategy:} The first global strategy that one might think of is to try to guess $P$, since passing the \texttt{test1} reduces to finding the strings that have bits `0' is all the $p_i$ positions i.e. positions marked $b=1$. If Eve successfully manages to guess the $b=1$ positions, then she has a deterministic strategy of winning the \texttt{test2}, since she also knows the $b=0$ trap positions. Across these positions she can deterministically assign the bits such that the \emph{count} of the number of 1 bits lie within the interval $\big \lvert count - \frac{N}{4} \big\rvert \leqslant \delta_{er}$.
  
We denotes Eve's generated string with this strategy to be $S_g$. Hence the probability of $S_g$ passing \texttt{test1} is equal to correctly guessing the $\frac{N}{2}$ positions marked $b=1$,
\begin{equation}
\text{Pr}[\texttt{test1 pass}_{S_g}] = \text{Pr}[\text{guess $b=1$ positions}] = {N\choose N/2}^{-1}
\end{equation}

Once this test passes, then test2 passes with certainty. Now the probability of passing the \texttt{cVer} verification algorithm is, 

\begin{equation}
\begin{split}
    \text{Pr}[\text{Ver accept}_{\text{Eve},S_g}] &= \text{Pr}[\texttt{test1 pass}_{S_g}\wedge \texttt{test2 pass}_{S_g}]\\
    &= \text{Pr}[\texttt{test1 pass}_{S_g}]\cdot \text{Pr}[\texttt{test2 pass}_{S_g}|\texttt{test1 pass}_{S_g}] \\
    &= {N\choose N/2}^{-1}\cdot 1\\
    &\leqslant N^{-\frac{N}{2}} 
\end{split}
\end{equation}
We show that this global strategy is not optimal and Eve can design an optimal global strategy by properly utilising the part the second part of the information.

First, we argue that maximising the number of 0's will necessarily increase the success probability of passing \texttt{test1}. Let us assume that Eve sends an all `0' string $S_g$ to Alice. Since \texttt{test1} checks only if in the $b=1$ marked positions are 0, so $S_g$ will always pass the first test. However, this string necessarily fails the \texttt{test2} since the \emph{count} for this test is $N/2$ which is much higher than the tolerated limit. 
 
Thus there always exists a global strategy with an optimal number of bits which are 1 in $S_g$ in the case of $\delta_{er}=0$, or more precisely a strategy that allows the flexibility of having a set of values for the number of `1' bits that the \texttt{test2} tolerates in case of $\delta_{er}\neq 0$. \\

\noindent \textbf{Optimal global strategy}: We say that an optimal global strategy $\E_{gop}$ is the one that outputs a string $S_{gop}$  with $c_1$ number of 1 bits, where $c_1 \in m_{valid} = \{\frac{N}{4} - \delta_{er}, \dots,  \frac{N}{4} + \delta_{er}\}$.\\

\noindent \textbf{Optimality argument}: We prove the optimality of our test by the contradiction argument. Let us assume that there is a strategy $\E_{g}$ different from above which produces a string $S_g$ that succeeds with the verification acceptance probability higher than $S_{gop}$. Now, either all the strings that $\E_g$ outputs have $c_1$ number of 1 bits, where $c_1$ lies within the optimal boundary $ m_{valid}$. In this case $\E_g$ falls within the $\E_{gop}$ strategy set. Or, there is at least one string that $\E_g$ outputs with $c_1$ number of 1 bits such that $c_1 \not\in m_{valid}$. In this case, that string will necessarily fail \texttt{test2}, even if it passes \texttt{test1}. This is because for the strategy $\E_g \not\in \E_{gop}$ to pass, the bits in $S_g$ which are 1 must necessarily appear in the positions marked $b=0$ (trap positions). And since the number of 1 bits $c_1 \not\in m_{valid}$, this implies it will fail the \texttt{test2}. Thus, $\text{Pr}[\text{Ver Accept}_{\text{Eve},\E_g \not\in \E_{gop}}] = 0$.

Note that the condition of $c_1 \in m_{valid}$ is necessary but not a sufficient condition for passing the verification algorithm $\texttt{cVer}$ i.e. any string with with $c_1 \not\in m_{valid}$, will always fail but not all strings with $c_1 \in m_{valid}$ will always pass the verification. Thus we can define the largest possible set of potentially valid strings which Eve needs to choose from to maximise her acceptance probability. As a result, we can define the optimal strategy $\E_{gop}$'s event space to be ${N\choose c_1}$. This is the set of all strings with the number of bits $c_1 \in m_{valid}$. We can now find the optimal global probability which is the probability that both the tests of \texttt{cVer} pass,

\begin{equation}
\begin{split}
    \text{Pr}[\text{Ver accept}_{\text{Eve}, S_{gop}}] &= \text{Pr}[\texttt{test1 pass}_{S_{gop}}\wedge \texttt{test2 pass}_{S_{gop}}]\\
    &= \text{Pr}[\texttt{test1 pass}_{S_{gop}}]\cdot \text{Pr}[\texttt{test2 pass}_{S_{gop}}|\texttt{test1 pass}_{S_{gop}}] 
\end{split}
\end{equation}
To calculate $\text{Pr}[\texttt{test1 pass}_{S_{gop}}]$, we need to find the number of strings $S_{gop}$ from the whole set of strings $\{0,1\}^N$ with $c_1 \in m_{valid}$ bits and which passes the first test. In other words, the string $S_g$ must have bits 0 in all the $b=1$ marked positions and the bits  1 in the $b=0$ marked positions.

Thus there are $N/2$ positions out $N$ where the bits 1 can be placed without the \texttt{test1} getting rejected.

For a specific $c_1$, the total number of such strings is equal to the possible ways of distributing $c_1$ objects (1's) in $N/2$ positions:
\begin{equation}
\#\text{(correct strings)} = {N/2 \choose c_1}    
\end{equation}

If one of these `correct strings' is picked, it will necessarily also satisfy the condition of the second test. Hence the conditional probability $\text{Pr}[\texttt{test2 pass}_{S_{gop}}|\texttt{test1 pass}_{S_{gop}}]  = 1$. The probability of passing the first test is,
\begin{equation}
    \text{Pr}[\texttt{test1 pass}_{S_{gop}}] = {N/2 \choose c_1}\Big/{N \choose c_1}
\end{equation}

The above \texttt{test1} passing probability is for a single $c_1 \in m_{valid}$. Summing over the probabilities of all the accepted $c_1$ ,
\begin{equation}
\text{Pr}[\texttt{test1 pass}_{S_{gop}}] = \sum_{c_1 \in m_{valid}} \frac{{\frac{N}{2}\choose c_1}}{{N\choose c_1}} = \sum^{\delta_{er}}_{k=-\delta_{er}} \frac{{\frac{N}{2}\choose \frac{N}{4} + k}}{{N\choose \frac{N}{4} + k}} = \frac{(\frac{N}{2})!}{N!} \sum^{\delta_{er}}_{k=-\delta_{er}} \frac{(\frac{3N}{4} - k)!}{(\frac{N}{4} - k)!}    
\end{equation}

In the limit $\delta_{er} \ll N$, the sum will converge,
\begin{equation}
\text{Pr}[\texttt{test1 pass}_{S_{gop}}] = (2\delta_{er} + 1)\cdot \frac{(\frac{N}{2})!(\frac{3N}{4})!}{N!(\frac{N}{4})!}    
\end{equation}

From the above equations, the probability that Eve passes the \texttt{cVer} algorithm  using the global strategy,

\begin{equation}
\begin{split}
\text{Pr}[\text{Ver accept}_{\text{Eve},S_{gop}}] 
&= \text{Pr}[\texttt{test1 pass}_{S_{gop}}]\cdot \text{Pr}[\texttt{test2 pass}_{S_{gop}}|\texttt{test1 pass}_{S_{gop}}] \\
&= (2\delta_{er} + 1)\times \frac{(\frac{N}{2})!(\frac{3N}{4})!}{N!(\frac{N}{4})!}\cdot 1 \\ 
&\leqslant \mathcal{O}(N^{-N/2})
\end{split}
\end{equation}

\noindent\textbf{3. Probability comparison of Independent guessing strategy and Global strategy:} To find the optimal classical attack, we compare the two categories of the attack strategies of Eve.

We fix the accepted tolerance value $\delta_{er}=1$ for the comparison. The same result holds for other fixed $\delta_{er}$ values. Figure~\ref{fig:Figure3} shows the behaviour of the acceptance probabilities of Eve in the independent guessing strategy and global strategy as an increasing function of the string length $N$.


From the simulation, we infer that the two strategies have inverse exponential form as expected. Also, they both converge for large enough $N$ values. This also confirms the fact that the optimal strategy lies in finding the correct number of 1's in the string and the difference comes from our approximation in using the frequency interpretation of the probabilities in the smaller N. Using Stirling's approximation $n! \approx \sqrt{2n\pi}(\frac{n}{e})^n$ one can check that $\frac{1}{{N\choose \frac{N}{4}}} \approx (\frac{4}{3^{3/4}})^{-N}$ which gives exactly the same bound as the independent guessing strategy. Although, in small $N$ the global strategy is slightly better.
Finally, we use Stirling's approximation ${2n\choose n} \approx \frac{2^{2n}}{\sqrt{\pi n}}$ to obtain the common factor of both probabilities we can bound the adversary's optimal success probability as,
\begin{equation}
\text{Pr}[\text{Ver Accept}_{\text{Eve}}] \approx \frac{3^{3N/4}}{2^{2N}}\times\frac{2^{N/2}}{\sqrt{\frac{\pi N}{4}}} = \frac{2}{\sqrt{N\pi}}(\frac{2^6}{{3^3}})^{-N/4} \approx \mathcal{O}(2^{-N}) \quad \text{for large enough N}    
\end{equation}

This completes the proof of Theorem~\ref{th:cv-clattack}.
%
%

\subsection{Average probability convergence}\label{ap:avgprob}
Here we approximate the following integral for the average probability that Eve wins the classical verification by performing the optimal classical strategy when $p$ is chosen to be a uniform distribution.
\[\underset{p}{\text{Pr}[\text{Ver accept}_{\text{Eve}}]} = \int_{0}^{1} \frac{2}{N + 2}\frac{(N-Np)!(\frac{N+Np}{2})!}{N!(\frac{N-Np}{2})!} dp\]

We choose $NP=k$ thus we have $Ndp=dk$ and we can rewrite the integral as:
\[\underset{p}{\text{Pr}[\text{Ver accept}_{\text{Eve}}]} =  \frac{2}{N(N + 2)}\int_{0}^{N}\frac{(N-k)!(\frac{N+k}{2})!}{N!(\frac{N-k}{2})!} dk\]

Now we can approximate the integral for discrete $k \in \{0,1,\dots,N\}$. Hence we have:
\[\underset{p}{\text{Pr}[\text{Ver accept}_{\text{Eve}}]} \approx \overline{\text{Pr}[\text{Ver accept}_{\text{Eve}}]} = \frac{2}{N(N + 2)} \sum^{N}_{k=0} \frac{(N-k)!(\frac{N+k}{2})!}{N!(\frac{N-k}{2})!}\]

The above series can be opened further as:

\begin{equation}
\begin{split}
    \sum^{N}_{k=0} \frac{(N-k)!(\frac{N+k}{2})!}{N!(\frac{N-k}{2})!} & = 1 + \frac{(N-1)!}{N!}\times\frac{(\frac{N}{2}+\frac{1}{2})!}{(\frac{N}{2}-\frac{1}{2})!} + \frac{(N-2)!}{N!}\times\frac{(\frac{N}{2}+1)!}{(\frac{N}{2}-1)!} + \dots + 1 \\
    & = 1 + \frac{1}{N}\times\frac{(\frac{N}{2} + \frac{1}{2})\cancel{(\frac{N}{2} - \frac{1}{2})!}}{\cancel{(\frac{N}{2} - \frac{1}{2})!}} + \frac{1}{N(N-1)}\times\frac{(\frac{N}{2}+1)(\frac{N}{2})\cancel{(\frac{N}{2}-1)!}}{\cancel{(\frac{N}{2}-1)!}} + \dots + 1 \\
    & \underset{N \gg 1}{\approx} 2 + \frac{\frac{N}{2}}{N} + \frac{(\frac{N}{2})^2}{N^2} + \frac{(\frac{N}{2})^3}{N^3} + \dots \\
    & = 2 + \sum^{N-1}_{i=1}(\frac{1}{2})^i \approx 2 + (1 - 2^{1-N}) \approx 3 
\end{split}
\end{equation}

where the sum has been approximated for large $N$. Thus we can write the average probability in the limit of large $N$ as follows,
\begin{equation}
\underset{p}{\text{Pr}[\text{Ver accept}_{\text{Eve}}]} \approx \overline{\text{Pr}[\text{Ver accept}_{\text{Eve}}]} = \frac{6}{N(N+2)}    
\end{equation}

\end{document}